\documentclass{article}
\usepackage{lmodern}

\usepackage[T1]{fontenc}
\usepackage[latin9]{inputenc}
\usepackage{geometry}
\usepackage[american,english]{babel}
\usepackage{array}
\usepackage{float}
\usepackage{amsmath}
\usepackage{amsthm}
\usepackage{amssymb}
\usepackage{xcolor}
\usepackage{graphicx}
\usepackage[numbers]{natbib}
\usepackage[unicode=true,pdfusetitle,
 bookmarks=true,bookmarksnumbered=false,bookmarksopen=false,
 breaklinks=true,pdfborder={0 0 0},pdfborderstyle={},backref=false,colorlinks=false]
 {hyperref}
\hypersetup{
 pdfborderstyle=,pdfusetitle=true}

\makeatletter

\providecommand{\tabularnewline}{\\}

\numberwithin{equation}{section}
\theoremstyle{plain}
\newtheorem{thm}{\protect\theoremname}[section]
\theoremstyle{plain}
\newtheorem{lem}[thm]{\protect\lemmaname}
\theoremstyle{remark}
\newtheorem{rem}[thm]{\protect\remarkname}
\theoremstyle{definition}
\newtheorem{example}[thm]{\protect\examplename}

\newtheorem{assumption}{Assumption}

\makeatother

\addto\captionsamerican{\renewcommand{\examplename}{Example}}
\addto\captionsamerican{\renewcommand{\lemmaname}{Lemma}}
\addto\captionsamerican{\renewcommand{\remarkname}{Remark}}
\addto\captionsamerican{\renewcommand{\theoremname}{Theorem}}
\addto\captionsenglish{\renewcommand{\examplename}{Example}}
\addto\captionsenglish{\renewcommand{\lemmaname}{Lemma}}
\addto\captionsenglish{\renewcommand{\remarkname}{Remark}}
\addto\captionsenglish{\renewcommand{\theoremname}{Theorem}}
\providecommand{\examplename}{Example}
\providecommand{\lemmaname}{Lemma}
\providecommand{\remarkname}{Remark}
\providecommand{\theoremname}{Theorem}

\begin{document}
\title{On the number of terms in the COS method for European option pricing }
\author{Gero Junike\thanks{Carl von Ossietzky Universität, Institut für Mathematik, 26129 Oldenburg,
Germany, ORCID: 0000-0001-8686-2661, E-mail: gero.junike@uol.de\\ \textcolor{blue}{This is a post-peer-review, pre-copyedit version of an article published in the Journal Numerische Mathematik.
The final authenticated version is available online at:
Junike, G. On the number of terms in the COS method for European option pricing Numerische Mathematik (2024). https://doi.org/10.1007/s00211-024-01402-1} }}
\maketitle
\begin{abstract}
The Fourier-cosine expansion (COS) method is used to price European
options numerically in a very efficient way. To apply the COS method,
one has to specify two parameters: a truncation range for the density
of the log-returns and a number of terms $N$ to approximate the truncated
density by a cosine series. How to choose the truncation range is
already known. Here, we are able to find an explicit and useful bound
for $N$ as well for pricing and for the sensitivities, i.e., the
Greeks Delta and Gamma, provided the density of the log-returns is
smooth. We further show that the COS method has an exponential order
of convergence when the density is smooth and decays exponentially.
However, when the density is smooth and has heavy tails, as in the
Finite Moment Log Stable model, the COS method does not have exponential
order of convergence. Numerical experiments confirm the theoretical
results.\\
 \\
 \textbf{Keywords: }COS method, option pricing, Greeks, number of
terms, order of convergence, heavy tails\\
 \textbf{Mathematics Subject Classification} 65D30 · 91B24 · 65T40 
\end{abstract}

\section{Introduction}

To calibrate stock price models, it is crucial to price European options
quickly because stock price models are typically calibrated to given
prices of liquid call and put options by minimizing the mean-square-error
between model prices and given market prices. During the optimization
routine, the model prices of call and put options need to be evaluated
often for different model parameters.

To compute the price of a European option, one must solve an integral
involving the product of the density of the log-returns at maturity
and the payoff function. However, for many financial models, the density
$f$ of the log-returns is unknown. Fortunately, the characteristic
function of the log-returns is often given in closed form and can
be used to obtain the density.

In their seminal paper, \citet{fang2009novel} proposed the COS method,
which is a very efficient way to approximate the density and to compute
option prices. The COS method requires two parameters: a truncation
range for the density of the log-returns and a number of terms $N$
to approximate the truncated density by a cosine series. While it
is known how to choose the truncation range, see \citet{junike2022precise},
the choice of $N$ is largely based on trial and error.

The COS method has been extensively extended and applied, see \citet{fang2009pricing,fang2011fourier,grzelak2011heston,ruijter2012two,zhang2013efficient,leitao2018data,liu2019neural,liu2019pricing,oosterlee2019mathematical,bardgett2019inferring}.
Other Fourier pricing techniques are discussed e.g., by \citet{carr1999option,lord2008fast,ortiz2013robust,ortiz2016highly}.

With respect to these papers we make the following three main contributions:
we develop an explicit, useful and rigorous bound for $N$; we analyze
the order of convergence of the COS method in detail; and we rigorously
analyze how the Greeks of an option can be approximated by the COS
method.

\citet{fang2009novel} propose to approximate the (unknown) density
in three steps: i) Truncate the density $f$, i.e., approximate $f$
by a function $f_{L}$ with finite support on some (sufficiently large)
interval $[-L,L]$. ii) Approximate $f_{L}$ by a Fourier-cosine expansion
$\sum a_{k}e_{k}^{L}$, where $a_{k}$ are Fourier coefficients of
$f_{L}$ and $e_{k}^{L}$ are cosine basis functions. iii) Approximate
$a_{k}$ by some coefficients $c_{k}$ which can be obtained directly
from the characteristic function of $f$. Thus, to apply the COS method,
two decisions must to be made: find a suitable truncation range $[-L,L]$
and identify the number $N$ of cosine functions.

One may apply a simple triangle inequality to bound the error of the
three approximations and obtain: 
\begin{equation}
\bigg\| f-\sideset{}{'}\sum_{k=0}^{N}c_{k}e_{k}^{L}\bigg\|_{2}\leq\big\Vert f-f_{L}\big\Vert_{2}+\Big\Vert f_{L}-\sideset{}{'}\sum_{k=0}^{N}a_{k}e_{k}^{L}\Big\Vert_{2}+\Big\Vert\sideset{}{'}\sum_{k=0}^{\infty}(a_{k}-c_{k})e_{k}^{L}\Big\Vert_{2}.\label{eq:error_bound}
\end{equation}
The first, second and third terms at the right-hand side of Inequality
(\ref{eq:error_bound}) correspond to approximations due to i), ii)
and iii), respectively.

It is well known that the series truncation error, i.e., the second
term on the right-hand side of Inequality (\ref{eq:error_bound}),
can be bounded using integration by parts, see \citet{boyd2001chebyshev}.
One contribution of this article is to use this idea in order to find
an explicit and useful bound for $N$, provided the density of the
log-returns is smooth. Our bound for $N$ is provably large enough
to ensure that the COS method converges within a predefined error
tolerance. There are many financial models with smooth densities having
semi-heavy tails. Examples include the Black-Scholes (BS) model, see
\citet{black1973pricing}, the Heston model, see \citet{heston1993closed},
the Normal Inverse Gaussian (NIG) model, see \citet{barndorff1997normal}
and the CGMY model with parameter $Y\in(0,1)$, see \citet{carr2002fine,albin2009asymptotic,kuchler2013tempered,asmussen2022role}.
The density of the log-returns in the Variance Gamma (VG) model, see
\citet{madan1998variance}, is not smooth for some parameters, and
our methodology cannot be applied to the VG model. We also compare
the solution for finding $N$ with another solution proposed by \citet{aimi2023fast}.

\citet{fang2009novel} also analyzed the order of convergence of the
COS method, focusing on the second term on the right-hand side of
Inequality (\ref{eq:error_bound}). They concluded that with a properly
chosen truncation range, the overall error converges exponentially
for smooth density functions and compares favorably to the Carr-Madan
formula, see \citet{carr1999option}. 

Another contribution of this article is to also consider the errors
introduced by the truncation range, i.e., the errors due to i), ii)
and iii), and to establish upper bounds for the order of convergence
of the COS method. We confirm, both theoretically and empirically,
that the COS method indeed converges exponentially for smooth density
functions if, in addition, the tails of the density decay at least
exponentially.

However, for fat-tailed and smooth densities, such as the density
of the log-returns in the Finite Moment Log Stable (FMLS) model (see
\citet{carr2003finite}), the truncation error due to i) and iii)
becomes much more relevant compared to densities with semi-heavy tails.
We show theoretically that the COS method converges at least as fast
as $O(N^{-\alpha})$ for $N\to\infty$, where $\alpha>0$ is the Pareto
tail index, e.g., for the FMLS model $\alpha\in(1,2)$. Empirical
experiments indicate that the COS method converges for such densities
as fast as $O(N^{-\alpha})$, i.e., the theoretical bound is sharp
and the COS method does not converge exponentially but the order of
convergence is $\alpha$.

\emph{Greeks}, also known as \emph{option} \emph{sensitivities}, play
an important role in risk management. The Greek letters Delta or Gamma
respectively represent the first and second derivatives of the price
of the option with respect to the current price of the underlying
asset. There are formulas in the literature on how to approximate
the Delta and Gamma of the option using the COS method, see \citep{fang2009novel,ruijter2015application,leitao2018data}.
Another contribution of this article is to provide explicit formulas
for the truncation range and the number of terms for the Greeks Delta
and Gamma.

This article is structured as follows: Section \ref{sec:Overview:-the-COS}
gives an overview of the technical details of the COS method. Section
\ref{sec:How-to-find} gives explicit formulas for the truncation
range and the number of terms. Section \ref{sec:Convergence-rate-of}
analyzes the order of convergence of the COS method. In Sections \ref{sec:How-to-find}
and \ref{sec:Convergence-rate-of}, we distinguish between models
with semi-heavy tails and models with heavy tails. Section \ref{sec:Greeks-by-the}
discusses the numerical computation of the Greeks using the COS method.
Section \ref{sec:Numerical-experiments} contains numerical experiments
that confirm the theoretical results. Section \ref{sec:Conclusions}
concludes.

\section{\label{sec:Overview:-the-COS}Overview: the COS method for option
pricing}

We model the stock price over time by a semimartingale $(S_{t})_{t\geq0}$
on a filtered probability space $(\Omega,\mathcal{F},P,(\mathcal{F}_{t})_{t\geq0})$.
The filtration $(\mathcal{F}_{t})_{t\geq0}$ satisfies the usual conditions
and $\mathcal{F}_{0}=\{\Omega,\emptyset\}$. We assume that there
is a bank account paying continuous compound interest $r\in\mathbb{R}$
and there is a risk-neutral measure $Q$. All expectations are taken
under $Q$. All densities are risk-neutral.

There is a European option with maturity $T>0$ and payoff $w(S_{T})$
at $T$, where $w:[0,\infty)\to\mathbb{R}$. For example, a European
put option with strike $K>0$ can be described by the payoff $w(x)=\max(K-x,0)$,
$x\geq0$.

In several places, we assume that the payoff function is bounded.
The prices of European call options are not bounded. If we want to approximate the 
price of a call option with a certain error tolerance, we need only approximate
the price of a put option within that error tolerance and apply the put-call parity.

Fix some $t_{0}\in[0,T)$. The price of the European option with payoff
$w$ at time $t_{0}$ is then given by 
\begin{equation}
e^{-r(T-t_{0})}E[w(S_{T})|\mathcal{F}_{t_{0}}].
\end{equation}
Since we only consider European options, we will focus on the time-0
price of the option and set $t_{0}=0$ for the remainder of the article. 

If we know the characteristic function $\varphi_{\log(S_{T})}$ of
$\log(S_{T})$ in closed form or we are able to obtain it numerically
efficiently, the COS method is able to price the European option numerically
very quickly, as follows: we denote by
\[
X_{T}:=\log(S_{T})-E[\log(S_{T})],\quad t\geq0
\]
the \emph{centralized log-returns.} The characteristic function $\varphi$
of $X_{T}$ is then equal to
\[
\varphi(u)=\varphi_{\log(S_{T})}(u)\exp(-iuE[\log(S_{T})]),\quad u\in\mathbb{R},
\]
where $E[\log(S_{T})]=-i\varphi_{\log(S_{T})}^{\prime}(0).$ We assume
that $X_{T}$ has a density $f$, but the exact structure of $f$
need not be known. Since $E[X_T]=0$, the density of $X_{T}$ is centered around
zero and it is justified to truncate the density $f$ on a symmetric
truncation range $[-L,L]$. Define 
\[
v(x):=e^{-rT}w(\exp(x+E[\log(S_{T})]),\quad x\in\mathbb{R}.
\]
The time-0 price of the European option with payoff $w$ is then given
by 
\begin{align}
e^{-rT}E[w(S_{T})] & =e^{-rT}\int_{\mathbb{R}}w(\exp(x+E[\log(S_{T})]))f(x)dx=\int_{\mathbb{R}}v(x)f(x)dx.\label{eq:price}
\end{align}

We need some abbreviations to discuss the COS method: suppose $f$
is $J+1$ times continuously differentiable for $J\geq0$. We will
approximate $f$ by cosine functions to solve the integral at the
right-hand side of Equation (\ref{eq:price}) numerically. We also
approximate the derivatives of $f$ by cosine functions in order to
approximate the Greeks, i.e., the sensitivities of the option, numerically.

By $f^{(j)}$ we denote the $j^{th}$-derivative of $f$. We use the
convention $f^{(0)}\equiv f$. For $L>0$ , let $f_{L}^{(j)}:=1_{[-L,L]}f^{(j)}$,
$j=0,...,J+1$. Suppose that $f^{(j)}$ is integrable and vanishes
at $\pm\infty$. By integration by parts, the Fourier transform of
$f^{(j)}$ is given by 
\begin{equation}
u\mapsto(-iu)^{j}\varphi(u),\quad u\in\mathbb{R},\quad j=0,...,J+1.\label{eq:Fourier_transofrm}
\end{equation}
Define the basis functions 
\[
e_{k}^{L}(x)=1_{[-L,L]}(x)\cos\left(k\pi\frac{x+L}{2L}\right),\quad x\in\mathbb{R},\quad k=0,1,...
\]
The Fourier coefficients of $f_{L}^{(j)}$ are defined by $a_{k}^{j}$
and approximated by $c_{k}^{j}$, where 
\begin{align*}
a_{k}^{j}:= & \frac{1}{L}\int_{-L}^{L}f^{(j)}(x)e_{k}^{L}(x)dx,\\
c_{k}^{j}:= & \frac{1}{L}\int_{\mathbb{R}}f^{(j)}(x)\cos\big(k\pi\frac{x+L}{2L}\big)dx,\quad k=0,1,...,\quad j=0,...,J+1.
\end{align*}
We also write $a_{k}$ and $c_{k}$ instead of $a_{k}^{0}$ and $c_{k}^{0}$,
respectively. Intuitively, we then have 
\[
f^{(j)}\approx f_{L}^{(j)}=\sum_{k=0}^{\infty}{}^{\prime}a_{k}^{j}e_{k}^{L}\approx\sum_{k=0}^{N}{}^{\prime}a_{k}^{j}e_{k}^{L}\approx\sum_{k=0}^{N}{}^{\prime}c_{k}^{j}e_{k}^{L},
\]
where $\sum{}^{\prime}$ indicates that the first summand (with $k=0$)
is weighted by one-half. A little analysis shows that 
\[
c_{k}^{j}=\frac{1}{L}\Re\bigg\{\bigg(-i\frac{k\pi}{2L}\bigg)^{j}\varphi\bigg(\frac{k\pi}{2L}\bigg)e^{i\frac{k\pi}{2}}\bigg\},\quad k=0,1,...,\quad j=0,...,J+1,
\]
i.e., the coefficients $c_{k}^{j}$ can be obtained explicitly if
$\varphi$ is given in closed form. Here, $\Re(z)$ denotes the real
part of a complex number $z$ and $i$ the imaginary unit. For $0<M\leq L$
define 
\begin{equation}
v_{k}:=\int_{-M}^{M}v(x)e_{k}^{L}(x)dx,\quad k=0,1,...\label{eq:vk}
\end{equation}
To keep the notation simple, we suppress the dependence of $a_{k}^{j}$
and $c_{k}^{j}$ on $L$ and the dependence of $v_{k}$ on $M$. The
COS method states that the time-0 price of the European option can
be approximated by 
\begin{equation}
\int_{\mathbb{R}}v(x)f(x)dx\approx\int_{-M}^{M}v(x)\sideset{}{'}\sum_{k=0}^{N}c_{k}e_{k}^{L}(x)dx=\sum_{k=0}^{N}{}^{\prime}c_{k}v_{k}.\label{eq:COS_method}
\end{equation}
The coefficients $c_{k}$ are given in closed form when $\varphi$
is given analytically and the coefficients $v_{k}$ can also be computed
explicitly in important cases, e.g., for plain vanilla European put
or call options and digital options, see \citet{fang2009novel}. This
makes the COS method numerically very efficient and robust.

In Lemma \ref{lem:cor8} we give the approximation in line (\ref{eq:COS_method})
a precise meaning. To do so, we need a bound for the term 
\begin{equation}
B_{f}(L):=\sum_{k=0}^{\infty}\frac{1}{L}\left|\int_{\mathbb{R}\setminus[-L,L]}f(x)\cos\left(k\pi\frac{x+L}{2L}\right)dx\right|^{2},\quad L>0.\label{eq:B(L)}
\end{equation}
Integrable functions $f$ with $B_{f}(L)\to0$, $L\to\infty$ are
called\emph{ COS-admissible.} The class of COS-admissible densities
is very large; in particular, it includes bounded densities with existing
first and second moments and stable densities, see \citet{junike2022precise}. 
\begin{lem}
\label{lem:cor8}Assume $f:\mathbb{R}\to\mathbb{R}$ is integrable
and square-integrable and COS-admissible. Let $v:\mathbb{R}\to\mathbb{R}$
be bounded, with $|v(x)|\le K$ for all $x\in\mathbb{R}$ and some
$K>0$. Let $\varepsilon>0$. Let $M>0$ so that 
\begin{equation}
\int_{\mathbb{R}\setminus[-M,M]}v(x)f(x)dx\leq\frac{\varepsilon}{2}.\label{eq:Cor8_1}
\end{equation}
Define $\xi=\sqrt{2M}K$. Let $L\geq M$ so that 
\begin{equation}
\left\Vert f-f_{L}\right\Vert _{2}\leq\frac{\varepsilon}{6\xi}\quad\text{and}\quad\sqrt{B_{f}(L)}\leq\frac{\varepsilon}{6\xi}.\label{eq:Cor8_2}
\end{equation}
Choose $N$ large enough so that 
\begin{equation}
\left\Vert f_{L}-\sum_{k=0}^{N}{}^{\prime}a_{k}e_{k}^{L}\right\Vert _{2}\leq\frac{\varepsilon}{6\xi}.\label{eq:series_trun_error}
\end{equation}
Then it follows that 
\[
\left|\int_{\mathbb{R}}v(x)f(x)dx-\sum_{k=0}^{N}{}^{\prime}c_{k}v_{k}\right|\leq\varepsilon.
\]
\end{lem}

\begin{proof}
\citet[Cor. 8]{junike2022precise}. 
\end{proof}
\begin{rem}
Often, it is fine to choose $M=L$, e.g., when applying the COS method
to densities with semi-heavy tails. However, if the density $f$ has
heavy tails, it is usually numerically more efficient to choose $L$
and $M$ differently.
\end{rem}

\section{\label{sec:How-to-find}On the choice of $N$ for smooth densities}

We summarize the assumptions about the density $f$ of the log-returns
in order to find explicit expressions for $M$, $L$ and $N$. We
denote by $C_{b}^{J+1}(\mathbb{R})$ the set of bounded functions
from $\mathbb{R}$ to $\mathbb{R}$ which are $(J+1)$-times, continuously
differentiable with bounded derivatives. By $\left\Vert .\right\Vert _{\infty}$
and $\left\Vert .\right\Vert _{2}$ we denote the supremum norm and
the $L^{2}$ norm, i.e., 
\[
\left\Vert f\right\Vert _{\infty}=\sup_{x\in\mathbb{R}}|f(x)|,\quad\left\Vert f\right\Vert _{2}=\sqrt{\int_{\mathbb{R}}(f(x))^{2}dx}.
\]

Let $\mathbb{N}_{0}=\{0\}\cup\mathbb{N}$ and $J\in\mathbb{N}_{0}$.
Let $C_{1}>0$, $C_{2}>0$ and $C_{3}>0$ be suitable constants. Let
$L_{0}>0$. Assume $f\in C_{b}^{J+1}(\mathbb{R})$. We say that $f$
\emph{and its derivatives have semi-heavy tails} if
\begin{equation}
|f^{(j)}(x)|\leq C_{1}C_{2}^{j}e^{-C_{2}|x|},\quad j=0,...,J+1,\quad|x|\geq L_{0}.\label{eq:A}
\end{equation}
We say that $f$ \emph{and its derivatives have heavy tails with index}
$\alpha>0$ if
\begin{equation}
|f^{(j)}(x)|\leq C_{3}\prod_{m=1}^{j}(\alpha+m)\,|x|^{-1-\alpha-j},\quad j=0,...,J+1,\quad|x|\geq L_{0}.\label{eq:B}
\end{equation}
We suppose that $f$ satisfies one of the following assumptions:

\begin{assumption}\label{A1} $f\in C_{b}^{J+1}(\mathbb{R})$ and
$f$ and its derivatives have semi-heavy tails.

\end{assumption}

\begin{assumption}\label{A2} $f\in C_{b}^{J+1}(\mathbb{R})$ and
$f$ and its derivatives have heavy tails with index $\alpha>0$.

\end{assumption} 
\begin{rem}
Assume the functions $u\mapsto|u^{j}\varphi(u)|$, $j=0,...,J+1$, are integrable. By Fourier inversion we have that $f(x)=\frac{1}{2\pi}\int_{\mathbb{R}}e^{iux}\varphi(u)du$.
By \citet[Lemma 2.8]{grubb2008distributions} it follows that $f\in C_{b}^{J+1}(\mathbb{R})$. 
\end{rem}

\begin{rem}
In exceptional cases, the constants $C_{1}$ and $C_{2}$ are explicitly
known; see Example \ref{exa:Lapalce}. However, it should be pointed
out that in Theorem \ref{thm:find_N_A_B} we obtain bounds for $M$,
$L$ and $N$ for models with semi-heavy tails (e.g., BS, VG, Heston,
NIG and CGMY) without knowing $C_{1}$ or $C_{2}$.
\end{rem}

\begin{rem}
In Theorem \ref{thm:find_N_A_B} we treat models with Pareto-tails;
i.e., the density for the log-returns behaves like 
\begin{equation}
f(x)\sim|x|^{-1-\alpha},\quad x\to\pm\infty.\label{eq:sym_Pareto}
\end{equation}
The right-hand side of Inequality (\ref{eq:B}) is obtained by differentiating
the right-hand side of (\ref{eq:sym_Pareto}). We assume in Theorem
\ref{thm:find_N_A_B} that $C_{3}$ and $\alpha$ are known. The exact
tail-behavior of the density of the log-returns is indeed known for
the stable law, in particular for the FMLS model. 
\end{rem}

\begin{example}
\label{exa:Lapalce}Let $\sigma>0$. In the Laplace model, see \citet{madan2016adapted},
the centralized log-returns at maturity $T>0$ are Laplace distributed
with variance $\sigma^{2}T$. To ensure stock prices are finite, we
need $\sigma\sqrt{T}<\sqrt{2}$, see \citet[Example 3]{guillaume2019implied}.
It holds that 
\[
|f_{\text{Lap}}^{(j)}(x)|=\frac{1}{\sqrt{2}\sigma\sqrt{T}}\left(\frac{\sqrt{2}}{\sigma\sqrt{T}}\right)^{j}e^{-\frac{\sqrt{2}}{\sigma\sqrt{T}}|x|},\quad x\in\mathbb{R}\setminus\{0\},\quad j=0,1,2,...
\]
Let $L_{0}>0$. Choose $C_{1}=\frac{1}{\sqrt{2}\sigma\sqrt{T}}$ and
$C_{2}=\frac{\sqrt{2}}{\sigma\sqrt{T}}$. Then $f_{\text{Lap}}^{(j)}$
satisfies Inequality (\ref{eq:A}). 
\end{example}

The following lemma makes it possible to bound the series-truncation
error, which depends only on the choice of $N$. It is known in a
similar form in the literature, see e.g., Theorem 1.39 in \citet{plonka2018numerical},
Theorem 4.2 in \citet{wright2015extension} and Theorem 6 in \citet{boyd2001chebyshev}.
It can be proven by integration by parts. It is usually stated for
functions with domain $[-1,1]$ or $[0,2\pi]$. Here, we explicitly
need the dependence of the series-truncation error on the truncation
range $[-L,L]$, so we give the full proof. 
\begin{lem}
\label{lem:A2}Let $J\in\mathbb{N}_{0}$. Suppose $f\in C_{b}^{J+1}(\mathbb{R})$.
It holds for $J\geq1$ that 
\begin{align}
\left\Vert f_{L}-\sum_{k=0}^{N}{}^{\prime}a_{k}e_{k}^{L}\right\Vert _{2}\leq & \sum_{j=1}^{J}\frac{2^{j+1}}{j\pi^{j+1}}\frac{L^{j+\frac{1}{2}}}{N^{j}}\left(|f^{(j)}(-L)|+|f^{(j)}(L)|\right)\nonumber \\
 & +\frac{2^{J+2}\|f^{(J+1)}\|_{\infty}}{J\pi^{J+1}}\frac{L^{J+\frac{3}{2}}}{N^{J}}\label{eq:trun_error}
\end{align}
and for $J=0$ that 
\[
\left\Vert f_{L}-\sum_{k=0}^{N}{}^{\prime}a_{k}e_{k}^{L}\right\Vert _{2}\leq\frac{4\|f^{(1)}\|_{\infty}}{\pi}\frac{L^{\frac{3}{2}}}{\sqrt{N}}.
\]
\end{lem}

\begin{proof}
It holds for any $\nu>0$ by integration by parts, see \citet[Eq. (1.3)]{lyness1971adjusted},
that 
\begin{align}
\int_{-L}^{L}f(x)e^{i\nu x}dx= & \sum_{j=0}^{J}\frac{i^{j+1}}{\nu^{j+1}}\left(e^{-i\nu L}f^{(j)}(-L)-e^{i\nu L}f^{(j)}(L)\right)\nonumber \\
 & +\frac{i^{J+1}}{\nu^{J+1}}\int_{-L}^{L}f^{(J+1)}(x)e^{i\nu x}dx.\label{eq:lyness}
\end{align}
For $k\in\mathbb{N}$, we apply Equation (\ref{eq:lyness}) for $\nu:=\frac{k\pi}{2L}$.
Then it follows that 
\begin{align*}
|a_{k}|= & \frac{1}{L}\left|\int_{-L}^{L}f(x)\cos\left(k\pi\frac{x+L}{2L}\right)dx\right|\\
= & \frac{1}{L}\left|\Re\left\{ e^{i\frac{k\pi}{2}}\int_{-L}^{L}f(x)e^{i\frac{k\pi}{2L}x}dx\right\} \right|\\
= & \frac{1}{L}\bigg|\Re\bigg\{\sum_{j=0}^{J}i^{j+1}\frac{(2L)^{j+1}}{(k\pi)^{j+1}}\left(f^{(j)}(-L)-(-1)^{k}f^{(j)}(L)\right)\\
 & +e^{i\frac{k\pi}{2}}i^{J+1}\frac{(2L)^{J+1}}{(k\pi)^{J+1}}\int_{-L}^{L}f^{(J+1)}(x)e^{i\frac{k\pi}{2L}x}dx\bigg\}\bigg|\\
\leq & \sum_{j=1}^{J}\frac{2^{j+1}}{\pi^{j+1}}\frac{L^{j}}{k^{j+1}}\left(|f^{(j)}(-L)|+|f^{(j)}(L)|\right)+\frac{2^{J+2}\left\Vert f^{(J+1)}\right\Vert _{\infty}}{\pi^{J+1}}\frac{L^{J+1}}{k^{J+1}}.
\end{align*}
Note that $\langle e_{k}^{L},e_{l}^{L}\rangle=L\delta_{k,l}$ for
$(k,l)\ne(0,0)$ and hence 
\begin{equation}
\left\Vert f_{L}-\sum_{k=0}^{N}{}^{\prime}a_{k}e_{k}^{L}\right\Vert _{2}=\sqrt{L\sum_{k=N+1}^{\infty}|a_{k}|^{2}}\leq\sqrt{L}\sum_{k=N+1}^{\infty}|a_{k}|.\label{eq:||f-ae||}
\end{equation}
By the integral test for convergence, 
\begin{equation}
\sum_{k=N+1}^{\infty}\frac{1}{k^{j+1}}\leq\int_{N}^{\infty}x^{-j-1}dx=\frac{1}{j}N^{-j},\quad j\geq1,\label{eq:sum1/k^j+1}
\end{equation}
which implies Eq. (\ref{eq:trun_error}) for $J\geq1$. If $J=0$,
apply the first Equality from (\ref{eq:||f-ae||}) and the Inequality
(\ref{eq:sum1/k^j+1}). 
\end{proof}
Given $L>0$, we need to find an upper bound for $\left\Vert f^{(j)}\right\Vert _{\infty}$
to estimate the series truncation error by Inequality (\ref{eq:trun_error}). It follows by the inverse Fourier transform and Equation (\ref{eq:Fourier_transofrm})
that 
\begin{equation}
\|f^{(j)}\|_{\infty}\leq\frac{1}{2\pi}\int_{\mathbb{R}}|u|^{j}|\varphi(u)|du,\quad j=0,1,...,J+1.\label{eq:H_j}
\end{equation}
Inequality (\ref{eq:H_j}) provides an explicit expression to find
a bound for the term $\|f^{(j)}\|_{\infty}$, $j=0,..,J+1$ for several
models. 
\begin{example}
\label{exa:NIG}In the symmetric NIG model with parameters $\alpha>0$,
$\beta=0$ and $\delta>0$, the centralized log-returns at time $T>0$
have density $f_{\text{NIG}}\in C_{b}^{\infty}(\mathbb{R})$, which
can be expressed in terms of the modified Bessel-function of the third
kind. The characteristic function is given by 
\[
u\mapsto\exp\left(-\delta T\sqrt{\alpha^{2}+u^{2}}+\delta T\alpha\right),\quad u\in\mathbb{R},
\]
see \citet{barndorff1997normal} and \citet[Sec. 5.3.8]{schoutens2003levy}.
We obtain, by Inequality (\ref{eq:H_j}) and using $\alpha^{2}+u^{2}\geq u^{2}$,
that 
\[
\|f_{\text{NIG}}^{(j)}\|_{\infty}\leq\frac{\exp(T\delta\alpha)j!}{(T\delta)^{j+1}\pi},\quad j=0,1,2,...
\]
\end{example}

We need Lemma \ref{lem:COS} to obtain a bound for $B_{f}(L)$. We
use the following abbreviation: the maximum of two real numbers $x,y$,
is denoted by $x\lor y$. 
\begin{lem}
\label{lem:COS}Assume $f\in C_{b}^{1}(\mathbb{R})$ is integrable
such that $xf^{2}(x)\to0$, $x\to\pm\infty$. Let $L>0$. If 
\begin{equation}
f^{\prime}(x)\geq0,\quad x\leq-L\quad\text{and}\quad f^{\prime}(x)\leq0,\quad x\geq L\label{eq:lemma_monotone}
\end{equation}
then 
\[
B_{f}(L)\leq\frac{1}{L}\left(\int_{\mathbb{R}\setminus[-L,L]}f(x)dx\right)^{2}+\frac{8}{3}L\left(f^{2}(L)\lor f^{2}(-L)\right).
\]
\end{lem}

\begin{proof}
It holds that $f(x)\to0$, $x\to\pm\infty$ and

\[
\int_{L}^{\infty}\left|f^{\prime}(x)\right|dx=-\int_{L}^{\infty}f^{\prime}(x)dx=f(L).
\]
Note that for all $k\in\mathbb{N}$, 
\begin{align*}
\int_{L}^{\infty}f(x)\cos\left(k\pi\frac{x+L}{2L}\right)dx= & \underbrace{\left[f(x)\sin\left(k\pi\frac{x+L}{2L}\right)\frac{2L}{k\pi}\right]_{L}^{\infty}}_{=0}\\
 & -\frac{2L}{k\pi}\int_{L}^{\infty}f^{\prime}(x)\sin\left(k\pi\frac{x+L}{2L}\right)dx,
\end{align*}
implying 
\begin{align*}
\frac{1}{L}\left|\int_{L}^{\infty}f(x)\cos\left(k\pi\frac{x+L}{2L}\right)dx\right|^{2}= & \frac{4}{\pi^{2}k^{2}}L\left|\int_{L}^{\infty}f^{\prime}(x)\sin\left(k\pi\frac{x+L}{2L}\right)dx\right|^{2}\\
\leq & \frac{4}{\pi^{2}k^{2}}Lf^{2}(L).
\end{align*}
Similarly, 
\[
\frac{1}{L}\left|\int_{-\infty}^{-L}f(x)\cos\left(k\pi\frac{x+L}{2L}\right)dx\right|^{2}\leq\frac{4}{\pi^{2}k^{2}}Lf^{2}(-L).
\]
Using 
\[
\sum_{k=1}^{\infty}\frac{1}{k^{2}}=\frac{\pi^{2}}{6}\quad\text{and}\quad|a+b|^{2}\leq4(a^{2}\lor b^{2}),
\]
we arrive at 
\begin{align*}
\sum_{k=0}^{\infty} & \frac{1}{L}\left|\int_{\mathbb{R}\setminus[-L,L]}f(x)\cos\left(k\pi\frac{x+L}{2L}\right)dx\right|^{2}\\
 & \leq\frac{1}{L}\left|\int_{\mathbb{R}\setminus[-L,L]}f(x)dx\right|^{2}+4\left(\frac{4\pi^{2}}{6\pi^{2}}Lf^{2}(L)\lor\frac{4\pi^{2}}{6\pi^{2}}Lf^{2}(-L)\right)\\
 & =\frac{1}{L}\left|\int_{\mathbb{R}\setminus[-L,L]}f(x)dx\right|^{2}+\frac{8}{3}L\left(f^{2}(L)\lor f^{2}(-L)\right).
\end{align*}
\end{proof}
In Theorem \ref{thm:find_N_A_B} we provide explicit formulas for
$M$, $N$ and $L$ when $f$ satisfies Assumption \ref{A1} or \ref{A2}
to ensure that the COS method approximates the true price within a
predefined error tolerance $\varepsilon>0$. We also include the derivatives
of $f$ in Theorem \ref{thm:find_N_A_B} in order to be able to approximate
the sensitivities (Greeks) of the price of the option, see Section
\ref{sec:Greeks-by-the}. To approximate the time-0 price of the option
by Theorem \ref{thm:find_N_A_B}, set $\ell=0$. 
\begin{thm}
\label{thm:find_N_A_B}(Find $M$, $L$ and $N$). Let $v:\mathbb{R}\to\mathbb{R}$
be bounded, with $|v(x)|\le K$ for all $x\in\mathbb{R}$ and some
$K>0$. Let $\varepsilon>0$ be small enough. Suppose $J\in\mathbb{N}_{0}$.\\
i) Assume density $f$ satisfies Assumption \ref{A1}. For some even
$n\in\mathbb{N}$ let 
\begin{equation}
L=M=\sqrt[n]{\frac{2K\mu_{n}}{\varepsilon}},\label{eq:M}
\end{equation}
where $\mu_{n}$ is the $n^{th}-$moment of $f$, i.e., $\mu_{n}=\frac{1}{i^{n}}\left.\frac{\partial^{n}}{\partial u^{n}}\varphi(u)\right|_{u=0}$.
Let $\xi=\sqrt{2M}K$. If $J\geq1$, let $\ell\in\{0,...,J-1\}$,
$k\in\{1,...,J-\ell\}$ and
\begin{align}
N & \geq\left(\frac{2^{k+2}\left\Vert f^{(k+1+\ell)}\right\Vert _{\infty}L^{k+\frac{3}{2}}}{k\pi^{k+1}}\frac{12\xi}{\varepsilon}\right)^{\frac{1}{k}}.\label{eq:N_J>0}
\end{align}
If $J=0$, let $\ell=0$ and
\begin{equation}
N\geq\left(\frac{4\left\Vert f^{(1)}\right\Vert _{\infty}L^{\frac{3}{2}}}{\pi}\frac{6\xi}{\varepsilon}\right)^{2}.\label{eq:N_j=00003D00003D0}
\end{equation}
ii) Assume density $f$ satisfies Assumption \ref{A2}, $J\geq1$
and $f$ is unimodal. Let $M=\left(\frac{4C_{3}K}{\varepsilon\alpha}\right)^{\frac{1}{\alpha}}$,
$\xi=\sqrt{2M}K$ and 
\begin{equation}
L=M\lor\left(12C_{3}\sqrt{\frac{1}{\alpha^{2}}+\frac{2}{3}}\frac{\xi}{\varepsilon}\right)^{\frac{2}{1+2\alpha}}.\label{eq:L_stable}
\end{equation}
Let $\ell=0$, $k\in\{1,...,J\}$ and define $N$ as in Equation (\ref{eq:N_J>0}).\\
In both cases, i) and ii), it holds that 
\begin{equation}
\left|\int_{\mathbb{R}}v(x)f^{(\ell)}(x)dx-\sum_{k=0}^{N}{}^{\prime}c_{k}^{\ell}v_{k}\right|\leq\varepsilon.\label{eq:vf-sum ck vk}
\end{equation}
\end{thm}

\begin{proof}
We start with case i). For $\varepsilon$ small enough, $M$ is large
enough so that Assumption \ref{A1} holds, i.e., $L_{0}\leq M=L$.
It follows that 
\begin{equation}
\int_{\mathbb{R}\setminus[-M,M]}\big|v(x)f^{(\ell)}(x)\big|dx\leq K\int_{\mathbb{R}\setminus[-M,M]}C_{1}C_{2}^{\ell}e^{-C_{2}|x|}(x)dx\leq\frac{\varepsilon}{2};\label{eq:A0_semi_eps}
\end{equation}
the last inequality holds true if 
\begin{equation}
M\geq-\frac{1}{C_{2}}\log\left(\frac{1}{KC_{1}C_{2}^{\ell-1}}\frac{\varepsilon}{4}\right).\label{eq:M_semi}
\end{equation}
Further, 
\begin{align}
\left\Vert f^{(\ell)}-f_{L}^{(\ell)}\right\Vert _{2} & \leq\sqrt{2C_{1}^{2}C_{2}^{2\ell}\int_{L}^{\infty}e^{-2C_{2}x}dx}\nonumber \\
 & =\frac{C_{1}C_{2}^{\ell}}{\sqrt{C_{2}}}e^{-C_{2}L}\label{eq:A1_semi}\\
 & \leq\frac{\varepsilon}{6\xi};\label{eq:A_1_semi_eps}
\end{align}
the last inequality holds true if 
\begin{equation}
L\geq-\frac{1}{C_{2}}\log\left(\frac{\sqrt{C_{2}}}{C_{1}C_{2}^{\ell}}\frac{\varepsilon}{6\xi}\right).\label{eq:L1_semi}
\end{equation}
Proposition 2 in \citet{junike2022precise} shows that 
\[
B_{f^{(\ell)}}(L)\leq\frac{2}{3}\frac{\pi^{2}}{L^{2}}\int_{\mathbb{R}\setminus[-L,L]}|xf^{(\ell)}(x)|^{2}dx.
\]
We then have 
\begin{align}
\sqrt{B_{f^{(\ell)}}(L)}\leq & \sqrt{\frac{4}{3}\frac{\pi^{2}C_{1}^{2}C_{2}^{2\ell}}{L^{2}}\int_{L}^{\infty}x^{2}e^{-2C_{2}x}dx}\nonumber \\
= & \frac{2\pi C_{1}C_{2}^{\ell}}{\sqrt{6C_{2}}}e^{-C_{2}L}\sqrt{1+\frac{1}{LC_{2}}+\frac{1}{2L^{2}C_{2}^{2}}}\label{eq:B(L)_semi}\\
\leq & \frac{\varepsilon}{6\xi};\label{eq:B(L)_semi_eps}
\end{align}
the last inequality holds true if 
\begin{equation}
L\geq-\frac{1}{C_{2}}\log\left(\left(\frac{2\pi C_{1}C_{2}^{\ell}}{\sqrt{6C_{2}}}\sqrt{1+\frac{1}{LC_{2}}+\frac{1}{2L^{2}C_{2}^{2}}}\right)^{-1}\frac{\varepsilon}{6\xi}\right).\label{eq:L_B_semi}
\end{equation}
Assume $J\geq1$. For $\varepsilon$ small enough, we have 
\begin{equation}
N\geq\frac{3LC_{2}}{\pi},\label{eq:N>=00003D00003D3LC2/pi}
\end{equation}
because the right-hand side of Inequality (\ref{eq:N>=00003D00003D3LC2/pi})
is of order $\varepsilon^{-\frac{1}{n}}$ while the right-hand side
of Inequality (\ref{eq:N_J>0}) is of order $\varepsilon^{-(\frac{1}{n}+\frac{2}{kn}+\frac{1}{k})}$.
By Inequality (\ref{eq:N>=00003D00003D3LC2/pi}) it follows that 
\begin{align*}
 & \sum_{j=1}^{k}\frac{2^{j+1}}{j\pi^{j+1}}\frac{L^{j+\frac{1}{2}}}{N^{j}}\left(|f^{(j+\ell)}(-L)|+|f^{(j+\ell)}(L)|\right)\\
\leq & \sum_{j=1}^{k}\frac{2^{j+1}}{j\pi^{j+1}}\frac{L^{j+\frac{1}{2}}}{N^{j}}\left(2C_{1}C_{2}^{j+\ell}e^{-C_{2}L}\right)\\
\leq & \frac{8}{\pi^{2}}C_{1}C_{2}^{1+\ell}e^{-C_{2}L}\frac{L^{\frac{3}{2}}}{N}\sum_{j=0}^{\infty}\left(\frac{2LC_{2}}{\pi N}\right)^{j}\\
\leq & \frac{8}{\pi^{2}}C_{1}C_{2}^{1+\ell}e^{-C_{2}L}\frac{\pi\sqrt{L}}{3C_{2}}\frac{1}{1-\frac{2LC_{2}}{\pi N}}\\
\leq & \frac{8C_{1}C_{2}^{\ell}}{\pi}e^{-C_{2}L}\sqrt{L}\leq\frac{\varepsilon}{12\xi};
\end{align*}
the last inequality holds true if 
\begin{equation}
L\geq-\frac{1}{C_{2}}\log\left(\frac{\pi}{8C_{1}C_{2}^{\ell}}\frac{\varepsilon}{12\xi\sqrt{L}}\right).\label{eq:L_N_semi}
\end{equation}
By Equation (\ref{eq:M}), $M$ and $L$ are of order $\varepsilon^{-\frac{1}{n}}$$.$
Hence, for $\varepsilon$ small enough, Inequalities (\ref{eq:M_semi}),
(\ref{eq:L1_semi}), (\ref{eq:L_B_semi}) and (\ref{eq:L_N_semi})
are indeed satisfied because the right-hand sides of these Inequalities
are of order $\log(\varepsilon)$. By the definition of $N$ in Inequality
(\ref{eq:N_J>0}), we also have 
\[
\frac{2^{k+2}\left\Vert f^{(k+1+\ell)}\right\Vert _{\infty}}{k\pi^{k+1}}\frac{L^{k+\frac{3}{2}}}{N^{k}}\leq\frac{\varepsilon}{12\xi}.
\]
By Lemma \ref{lem:A2} it follows that 
\begin{equation}
\left\Vert f_{L}^{(\ell)}-\sum_{k=0}^{N}{}^{\prime}a_{k}^{\ell}e_{k}^{L}\right\Vert _{2}\leq\frac{\varepsilon}{6\xi}.\label{eq:<lesp/6xi}
\end{equation}
As $B_{f^{(\ell)}}(L)\to0$, $L\to\infty$, $f^{(\ell)}$ is COS-admissible.
Inequalities (\ref{eq:A0_semi_eps}), (\ref{eq:A_1_semi_eps}), (\ref{eq:B(L)_semi_eps}),
(\ref{eq:<lesp/6xi}) and Lemma \ref{lem:cor8} imply Inequality (\ref{eq:vf-sum ck vk}).
Now assume $J=0$. By the definition of $N$ in Inequality (\ref{eq:N_j=00003D00003D0}),
Lemma \ref{lem:A2} again implies Inequality (\ref{eq:<lesp/6xi}).
Apply Lemma \ref{lem:cor8} to conclude. 

Next, we treat the case ii). For $\varepsilon$ small enough, $M$
is large enough so that Assumption \ref{A2} holds, i.e., $L_{0}\leq M\leq L$.
The inequality 
\[
\int_{\mathbb{R}\setminus[-M,M]}\big|v(x)f(x)\big|dx\leq K\int_{\mathbb{R}\setminus[-M,M]}f(x)dx\leq\frac{2KC_{3}}{\alpha}M^{-\alpha}\leq\frac{\varepsilon}{2}
\]
is satisfied by the definition of $M$. A little calculation shows
that the definition of $L$ implies 
\[
L\geq\left(\frac{\sqrt{2}6C_{3}\xi}{\sqrt{1+2\alpha}\varepsilon}\right)^{\frac{2}{1+2\alpha}},
\]
therefore 
\begin{align}
\big\Vert f-f_{L}\big\Vert_{2} & \leq\sqrt{2\int_{L}^{\infty}C_{3}^{2}x^{-2-2\alpha}dx}\nonumber \\
 & =\frac{\sqrt{2}C_{3}}{\sqrt{1+2\alpha}}L^{-\frac{1+2\alpha}{2}}\label{eq:B(L)_stable-1}\\
 & \leq\frac{\varepsilon}{6\xi}.\nonumber 
\end{align}
Since $f$ is a unimodal density satisfying Assumption \ref{A2},
$f$ also satisfies the assumption of Lemma \ref{lem:COS}. To see
this, note that the unimodality implies the assertion in line (\ref{eq:lemma_monotone}).
Further, 
\[
|xf^{2}(x)|\leq C_{3}^{2}|x|^{-1-2\alpha}\to0,\quad x\to\pm\infty.
\]
Using the bound for $B_{f}(L)$ from Lemma \ref{lem:COS}, we obtain
\begin{align}
\sqrt{B_{f}(L)} & \leq\sqrt{\frac{1}{L}\left(2\int_{L}^{\infty}C_{3}x^{-1-\alpha}dx\right)^{2}+\frac{8}{3}LC_{3}^{2}L^{-2\alpha-2}}\nonumber \\
 & =\sqrt{\frac{4C_{3}^{2}}{\alpha^{2}}L^{-2\alpha-1}+\frac{8}{3}C_{3}^{2}L^{-2\alpha-1}}\nonumber \\
 & =2C_{3}\sqrt{\frac{1}{\alpha^{2}}+\frac{2}{3}}L^{-\frac{1+2\alpha}{2}}\label{eq:B(L)_stable}\\
 & \leq\frac{\varepsilon}{6\xi},\nonumber 
\end{align}
the last Inequality holds by the definition of $L$. For $\varepsilon$
small enough, $N$ is large enough and we have 
\begin{equation}
\frac{\alpha+k}{N-\frac{2(\alpha+k)}{\pi}}\leq\sqrt{\frac{1}{\alpha^{2}}+\frac{2}{3}}.\label{eq:1/N}
\end{equation}
Use Inequality (\ref{eq:1/N}), the definition of $L$ and $\frac{96}{\pi^{2}}\leq12$,
to see that 
\begin{equation}
L\geq\left(\frac{96C_{3}\xi}{\pi^{2}\varepsilon}\frac{\alpha+k}{N-\frac{2(\alpha+k)}{\pi}}\right)^{\frac{2}{1+2\alpha}}.\label{eq:1/N_L}
\end{equation}
Using $\prod_{m=1}^{j}(\alpha+m)\leq(\alpha+k)^{j}$ for $j\leq k$,
it follows that 
\begin{align*}
 & \sum_{j=1}^{k}\frac{2^{j+1}}{j\pi^{j+1}}\frac{L^{j+\frac{1}{2}}}{N^{j}}\left(|f^{(j)}(-L)|+|f^{(j)}(L)|\right)\\
\leq & \sum_{j=1}^{k}\frac{2^{j+1}}{j\pi^{j+1}}\frac{L^{j+\frac{1}{2}}}{N^{j}}\left(2C_{3}(\alpha+k)^{j}L^{-1-\alpha-j}\right)\\
\leq & \frac{8}{\pi^{2}}C_{3}L^{-\frac{1}{2}-\alpha}\frac{\alpha+k}{N}\sum_{j=0}^{\infty}\left(\frac{2(\alpha+k)}{N\pi}\right)^{j}\\
\leq & \frac{8}{\pi^{2}}C_{3}L^{-\frac{1}{2}-\alpha}\frac{\alpha+k}{N-\frac{2(\alpha+k)}{\pi}}\\
\leq & \frac{\varepsilon}{12\xi},
\end{align*}
the last inequality holds by Inequality (\ref{eq:1/N_L}). By the
definition of $N$, we also have 
\[
\frac{2^{k+2}\left\Vert f^{(k+1)}\right\Vert _{\infty}}{k\pi^{k+1}}\frac{L^{k+\frac{3}{2}}}{N^{k}}\leq\frac{\varepsilon}{12\xi}.
\]
By Lemma \ref{lem:A2} it follows that 
\[
\left\Vert f_{L}-\sum_{k=0}^{N}{}^{\prime}a_{k}e_{k}^{L}\right\Vert _{2}\leq\frac{\varepsilon}{6\xi}.
\]
As $B_{f}(L)\to0$, $L\to\infty$, $f$ is COS-admissible. Apply Lemma
\ref{lem:cor8} to conclude. 
\end{proof}
\begin{rem}
\label{rem:Theorem_A}If $v$ is a European put option with maturity
$T$, $K$ can be set to the strike of the option times $e^{-rT}$.
The error tolerance is described by $\varepsilon$. Numerical experiments
in \citet{junike2022precise} suggest choosing $n\in\{4,6,8\}$ for
$\mu_{n}$. According to Theorem \ref{thm:find_N_A_B}, any $k\in\{1,...,J-\ell\}$
is allowed to define $N$ by Inequality (\ref{eq:N_J>0}). In the
applications, one could minimize $N$ over all admissible $k$. But
this could be time-consuming, and in the applications, we set $k$
to a fixed value, e.g., for the BS model, $k=40$ is a reasonable
choice, see Section \ref{subsec:BS-model}. For other models, another
choice for $k$ might be more suitable. Bounds for $\|f^{(k+1)}\|_{\infty}$
are explicitly known for some models, e.g., the BS, NIG and FMLS models.
These bounds can also be estimated numerically, e.g., for the Heston
model. Section \ref{sec:Numerical-experiments} contains examples
indicating that the bound for $N$ is often sharp and very useful
in applications. 
\end{rem}

\begin{rem}
If a density satisfies Assumption \ref{A1}, it also satisfies Assumption
\ref{A2}, i.e., theoretically, case i) in Theorem \ref{thm:find_N_A_B}
is included in case ii). However, in i) we do not need to know the
exact tail behavior of the density, i.e., the constants $C_{1}$ and
$C_{2}$ from Assumption \ref{A1}, in order to estimate the truncation
range because we apply Markov's inequality to find a bound for $M$.
This approach is not applicable to densities with heavy tails because
higher moments usually do not exist. In ii), we assume the tail behavior
of the density is known precisely, i.e., we have to know $C_{3}$
and $\alpha$ in Assumption \ref{A2} to estimate $M$, $L$ or $N$.
The constants $C_{3}$ and $\alpha$ are known, for example, for the
FMLS model.
\end{rem}

\section{\label{sec:Convergence-rate-of}Order of Convergence}

Theorem \ref{thm:asy} describes the order of convergence of the COS
method if we allow $N\to\infty$ and choose $M$ and $L$ depending
on $N$. We describe only the asymptotic behavior of the COS method
and we assume $M=L$ in this section to keep the notation simple.
We establish a bound of the order of convergence of the error of the
COS method with parameters $L$ and $N$, i.e., 
\[
\text{err}(L,N)=\left|\int_{\mathbb{R}}v(x)f(x)dx-\sum_{k=0}^{N}{}^{\prime}c_{k}v_{k}\right|.
\]
Let $M=L=L(N)$. We say the error of the COS method converges \emph{with
order $\rho>0$}, if there is a constant $\kappa>0$ such that for
all $N\in\mathbb{N}$ it holds that 
\begin{equation}
\text{err}(L(N),N)\leq\frac{\kappa}{N^{\rho}}.\label{eq:order of con}
\end{equation}
The error is of \emph{infinite order} or \emph{converges exponentially},
if Inequality (\ref{eq:order of con}) holds for any $\rho$, see
\citet[Sec. 2.3]{boyd2001chebyshev}. We use big $O$ notation: for
functions $g,h:\mathbb{N}\to\mathbb{R}$, we write $h(N)=O(g(N))$
as $N\to\infty$, if the absolute value of $h(N)$ is at most a positive
constant multiple of $g(N)$ for all sufficiently large values of
$N$. 
\begin{thm}
\label{thm:asy}(Bounds for the order of convergence). Let $v:\mathbb{R}\to\mathbb{R}$
be bounded, with $|v(x)|\le K$ for all $x\in\mathbb{R}$ and some
$K>0$. Assume $J\in\mathbb{N}$. \\
(i) Assume density $f$ satisfies Assumption \ref{A1}. Let $\beta\in\left(0,\frac{J}{J+3}\right)$
and $\gamma>0$. If $M=L=\gamma N^{\beta}$ it holds that 
\[
\text{err}(L(N),N)\leq O\left(N^{-J(1-\beta)+2\beta}\right),\quad\text{as }N\to\infty.
\]
(ii) Assume density $f$ is unimodal and satisfies Assumption \ref{A2}.
Let $\beta\in\left(0,\frac{J}{J+2+2\alpha}\right)$ and $\gamma>0$.
If $M=L=\gamma N^{\beta}$ it holds that 
\[
\text{err}(L(N),N)\leq O\left(N^{-\alpha\beta}\right),\quad\text{as }N\to\infty.
\]
In both cases we have $\text{err}(L(N),N)\to0$, $N\to\infty.$ 
\end{thm}

\begin{proof}
Let $v_{L}:=1_{[-L,L]}v$. As in the proof of Corollary 8 in \citet{junike2022precise},
one can show that 
\begin{align}
\Big|\int_{\mathbb{R}}v(x)f(x)dx-\sideset{}{'}\sum_{k=0}^{N}c_{k}v_{k}\Big|\leq & A_{0}(L)+\|v_{L}\|_{2}\left(A_{1}(L)+A_{2}(L,N)+\sqrt{B_{f}(L)}\right),\label{eq:asy_sum}
\end{align}
where $B_{f}(L)$ as in Equation (\ref{eq:B(L)}) and 
\begin{align*}
A_{0}(L)= & \int_{\mathbb{R}\setminus[-L,L]}\big|v(x)f(x)\big|dx,\\
A_{1}(L)= & \left\Vert f-f_{L}\right\Vert _{2},\\
A_{2}(L,N)= & \big\Vert f_{L}-\sideset{}{'}\sum_{k=0}^{N}a_{k}e_{k}^{L}\big\Vert_{2}.
\end{align*}
We will state upper bounds for $A_{0}$, $A_{1}$ and $B_{f}$ depending
on the tail behaviour of $f$, i.e., for the different cases i) and
ii) in Theorem \ref{thm:asy}. An upper bound for $A_{2}$ can be
obtained from Lemma \ref{lem:A2}. Note that 
\[
\|v_{L}\|_{2}\leq\sqrt{2L}K=\sqrt{2\gamma}KN^{\frac{\beta}{2}}.
\]
We now prove (i). For $L$ large enough, we have 
\[
A_{0}(L)\leq2KC_{1}\int_{L}^{\infty}e^{-C_{2}x}dx=\frac{2KC_{1}}{C_{2}}e^{-C_{2}L}.
\]
Further, by Inequality (\ref{eq:A1_semi}), it holds that 
\begin{align*}
A_{1}(L) & \leq\frac{C_{1}}{\sqrt{C_{2}}}e^{-C_{2}L}.
\end{align*}
Assuming $L\geq1$ and applying Inequality (\ref{eq:B(L)_semi}),
it follows that 
\begin{align*}
\sqrt{B_{f}(L)}\leq & \frac{2\pi C_{1}}{\sqrt{6C_{2}}}\sqrt{1+\frac{1}{C_{2}}+\frac{1}{2C_{2}^{2}}}e^{-C_{2}L}.
\end{align*}
Inequality (\ref{eq:trun_error}) implies 
\begin{align*}
A_{2}(L,N) & \leq\sum_{j=1}^{J}\frac{2^{j+1}}{j\pi^{j+1}}\frac{L^{j+\frac{1}{2}}}{N^{j}}\left(2C_{1}C_{2}^{j}e^{-C_{2}L}\right)+\frac{2^{J+2}\left\Vert f^{(J+1)}\right\Vert _{\infty}}{J\pi^{J+1}}\frac{L^{J+\frac{3}{2}}}{N^{J}}\\
 & \leq e^{-C_{2}L}\sqrt{L}\sum_{j=1}^{J}\frac{2^{j+2}C_{1}C_{2}^{j}\gamma^{j}}{j\pi^{j+1}}+\frac{2^{J+2}\left\Vert f^{(J+1)}\right\Vert _{\infty}}{J\pi^{J+1}}\frac{L^{J+\frac{3}{2}}}{N^{J}},
\end{align*}
where we used $\frac{L}{N}\leq\gamma$.

Let $b_{1},...,b_{4}>0$ be suitable constants. By Inequality (\ref{eq:asy_sum}),
it follows for $N$ large enough that 
\begin{align}
\Big|\int_{\mathbb{R}}v(x)f(x)dx-\sideset{}{'}\sum_{k=0}^{N}c_{k}^{L}v_{k}^{M}\Big|\leq & b_{1}e^{-C_{2}\gamma N^{\beta}}+b_{2}N^{\frac{\beta}{2}}\big(b_{3}N^{\frac{\beta}{2}}e^{-C_{2}\gamma N^{\beta}}+b_{4}N^{\beta(J+\frac{3}{2})-J}\big)\nonumber \\
\leq & O\left(N^{-J(1-\beta)+2\beta}\right),\quad N\to\infty.\label{eq:asy_cos_semi}
\end{align}
$\beta<\frac{J}{J+3}$ implies $-J(1-\beta)+2\beta<0$ and the right-hand
side of (\ref{eq:asy_cos_semi}) converges to zero for $N\to\infty$.
We show (ii). It holds for $L$ large enough using Inequalities (\ref{eq:B(L)_stable-1})
and (\ref{eq:B(L)_stable}) that 
\begin{align*}
A_{0}(L) & \leq2KC_{3}\int_{L}^{\infty}x^{-1-\alpha}dx\leq\frac{2KC_{3}}{\alpha}L^{-\alpha},\\
A_{1}(L) & \leq\frac{\sqrt{2}C_{3}}{\sqrt{1+2\alpha}}L^{-\frac{1}{2}-\alpha},\\
\sqrt{B_{f}(L)} & \leq2C_{3}\sqrt{\frac{1}{\alpha^{2}}+\frac{2}{3}}L^{-\frac{1}{2}-\alpha}.
\end{align*}
Inequality (\ref{eq:trun_error}) and Assumption \ref{A2} imply for
$N$ large enough and for some suitable constants $a_{1},...,a_{J}>0$
that 
\begin{align*}
A_{2}(L,N)\leq & \sum_{j=1}^{J}a_{j}N^{\beta(-\frac{1}{2}-\alpha)-j}+\frac{2^{J+2}\left\Vert f^{(J+1)}\right\Vert _{\infty}}{J\pi^{J+1}}\frac{L^{J+\frac{3}{2}}}{N^{J}}.
\end{align*}
By Inequality (\ref{eq:asy_sum}), it holds for some suitable constants
$b_{1},...,b_{4}>0$ that 
\begin{align*}
 & \Big|\int_{\mathbb{R}}v(x)f(x)dx-\sideset{}{'}\sum_{k=0}^{N}c_{k}v_{k}\Big|\\
 & \leq b_{1}N^{-\alpha\beta}+b_{2}N^{\frac{\beta}{2}}\bigg(b_{3}N^{-\beta(\alpha+\frac{1}{2})}+\sum_{j=1}^{J}a_{j}N^{\beta(-\frac{1}{2}-\alpha)-j}+b_{4}N^{\beta(J+\frac{3}{2})-J}\bigg)\\
 & \leq b_{1}N^{-\alpha\beta}+b_{2}\bigg(b_{3}N^{-\beta\alpha}+\sum_{j=1}^{J}a_{j}N^{-\beta\alpha-j}+b_{4}N^{\beta\left(J+2\right)-J}\bigg)\\
 & \leq O\left(N^{-\alpha\beta}\right),\quad N\to\infty.
\end{align*}
The last inequality can be seen as follows: as $\beta\leq\frac{J}{J+2+\alpha}$,
it follows that 
\[
-\alpha\beta\geq\beta(J+2)-J.
\]
\end{proof}
\begin{rem}
The COS method converges exponentially, i.e., faster than $O(N^{-\rho})$
as $N\to\infty$ for any $\rho>0$ if $f$ is smooth and has semi-heavy
tails. To see this, let $0<\beta<1$: then, 
\[
-J(1-\beta)+2\beta\to-\infty,\quad J\to\infty.
\]
\end{rem}

\begin{rem}
By Theorem \ref{thm:asy}, the COS method converges at least like
$O(N^{-\alpha})$ as $N\to\infty$ if $f$ is smooth and has heavy
tails with Pareto index $\alpha>0$. Numerical experiments indicate
that the COS method does not converge faster than $O(N^{-\alpha})$
for the FMLS model, see Section \ref{subsec:FMLS-model-1}, i.e.,
the bound in Theorem \ref{thm:asy}(ii) is sharp. 
\end{rem}

\begin{rem}
Theorem \ref{thm:asy} cannot be applied to densities that are non-differentiable,
e.g., the density of the VG model is non-differentiable for some parameters.
To improve the order of convergence of the COS method if the density
of the log-returns is non-differentiable, \citet{ruijter2015application}
apply spectral filters and consider the \emph{filter-COS method,}
$\sum_{k=0}^{N}{}^{\prime}\hat{s}(\frac{k}{N})c_{k}v_{k}$, where
$\hat{s}$ is a spectral filter, i.e., a smooth function with support
$[-1,1]$ and $\hat{s}(0)=1$. For an analysis of the order of convergence
and some error bounds for the filter-COS method, see \citet{ruijter2015application}
and references therein.
\end{rem}

\section{\label{sec:Greeks-by-the}On the Greeks}

The \emph{Greeks }or \emph{sensitivities }of a European option play
an important role in hedging and risk management. The most important
Greeks are Delta and Gamma, which are the first and second derivative
of the price of a European option with respect to the current underlying
 price $S_{0}>0$.

\citet[Remark 3.2]{fang2009novel} state formulas for the approximation
of Delta and Gamma by the COS method. We proof these formulas and
discuss how to choose $M$, $L$ and $N$ for the Greeks.

In this section we assume $S_{t}=S_{0}\bar{S}_{t}$ for some stochastic
process $(\bar{S}_{t})_{t\geq0}$, which does not depend on $S_{0}$
anymore. This assumption is a very typical one, see \citet[Sec. 3.1.2]{madan2016applied}.
As in Section \ref{sec:Overview:-the-COS}, we consider a European
option with maturity $T>0$ and payoff $w(S_{T})$ for some $w:[0,\infty)\to\mathbb{R}$.
Let $\eta:=E[\log(S_{T})]-\log(S_{0})$. Then $\eta$ does not depend
on $S_{0}$. Define

\begin{equation}
v(x,s):=e^{-rT}w(\exp(x+\log(s)+\eta)),\quad x\in\mathbb{R},\quad s>0.\label{eq:v_Greek}
\end{equation}
The time-0 price of the European option is then given by 
\begin{align*}
e^{-rT}E[w(S_{T})] & =\int_{\mathbb{R}}v(x,S_{0})f(x)dx,
\end{align*}
where, as before, $f$ is the density of the centralized log-returns.
Delta and Gamma are defined by 
\begin{equation}
\frac{\partial^{\ell}}{\partial S_{0}^{\ell}}\int_{\mathbb{R}}v(x,S_{0})f(x)dx,\quad\ell=1,2\label{eq:DeltaGamma}
\end{equation}
if the partial derivatives exist. The next lemma provides some conditions
to interchange integration and differentiation in Equation (\ref{eq:DeltaGamma}). 
\begin{lem}
\label{lem:diff_lemma}Let $w$ be bounded. Let $v$ be defined as
in Equation (\ref{eq:v_Greek}). Assume $J\in\mathbb{N}_{0}$ and
density $f$ satisfies Assumption \ref{A1}. It follows that

\[
\frac{\partial}{\partial S_{0}}\int_{\mathbb{R}}v(x,S_{0})f(x)dx=-\frac{1}{S_{0}}\int_{\mathbb{R}}v(x,S_{0})f^{(1)}(x)dx
\]
and if $J\geq1$ 
\[
\frac{\partial^{2}}{\partial S_{0}^{2}}\int_{\mathbb{R}}v(x,S_{0})f(x)dx=\frac{1}{S_{0}^{2}}\int_{\mathbb{R}}v(x,S_{0})\big(f^{(1)}(x)+f^{(2)}(x)\big)dx.
\]
\end{lem}

\begin{proof}
Let $K>0$ such that $v$ is bounded by $K$. Let $\underline{s},\overline{s}\in\mathbb{R}$
such that $0<\underline{s}<S_{0}<\overline{s}$. Let 
\[
h(x,s):=v(x,1)f(x-\log(s)),\quad(x,s)\in\mathbb{R}\times(\underline{s},\overline{s}).
\]
Then $x\mapsto h(x,s)$ is integrable for all $s\in(\underline{s},\overline{s})$
and the partial derivative 
\[
\frac{\partial}{\partial s}h(x,s)=-\frac{v(x,1)}{s}f^{(1)}(x-\log(s))
\]
exists for all $(x,s)\in\mathbb{R}\times(\underline{s},\overline{s})$.
Define $x_{0}:=L_{0}+|\log(\overline{s})|+|\log(\underline{s})|$.
Then, 
\[
|x-\log(s)|\geq L_{0},\quad(x,s)\in\mathbb{R}\setminus[-x_{0},x_{0}]\times(\underline{s},\overline{s}).
\]
Let 
\begin{align*}
m(x) & :=\begin{cases}
K\underline{s}^{-1}\left\Vert f^{(1)}\right\Vert _{\infty} & ,x\in[-x_{0},x_{0}]\\
K\underline{s}^{-C_{2}-1}C_{1}C_{2}e^{C_{2}x} & ,x<-x_{0}\\
K\underline{s}^{-1}\overline{s}^{C_{2}}C_{1}C_{2}e^{-C_{2}x} & ,x>x_{0}.
\end{cases}
\end{align*}
Then $|\frac{\partial}{\partial s}h(x,s)|\leq m(x)$ for all $(x,s)\in\mathbb{R}\times(\underline{s},\overline{s})$
and $x\mapsto m(x)$ is integrable. Interchanging differentiation
and integration is allowed by the dominated convergence theorem, see
e.g., \citet[Lemma 2.8]{grubb2008distributions}, and it follows that
\begin{align*}
\frac{\partial}{\partial S_{0}}\int_{\mathbb{R}}v(x,S_{0})f(x)dx & =\frac{\partial}{\partial S_{0}}\int_{\mathbb{R}}v(x,1)f(x-\log(S_{0}))dx\\
 & =-\frac{1}{S_{0}}\int_{\mathbb{R}}v(x,1)f^{(1)}(x-\log(S_{0}))dx\\
 & =-\frac{1}{S_{0}}\int_{\mathbb{R}}v(x,S_{0})f^{(1)}(x)dx.
\end{align*}
If $J\geq1$, $f$ is twice differentiable. Apply the arguments above
to $f^{(1)}$ to conclude. 
\end{proof}
In Theorem \ref{thm:Greeks} we provide explicit formulas for $M$,
$N$ and $L$ when $f$ satisfies Assumption \ref{A1} to ensure that
the COS method approximates the time-0 price and the Greeks Delta
and Gamma within a predefined error tolerance $\varepsilon>0$. One
can use the same parameters $M$, $N$ and $L$ to obtain both the
price and the Greeks. We define $v_{k}:=\int_{-M}^{M}v(x,S_{0})e_{k}^{L}(x)dx$
as in Equation (\ref{eq:vk}). 
\begin{thm}
\label{thm:Greeks}($M$, $L$ and $N$ for the time-0 price, Delta
and Gamma). Let $w$ be bounded. Let $v$ be defined as in Equation
(\ref{eq:v_Greek}). Let $\varepsilon>0$ be small enough. Let $\gamma=\min\big\{\varepsilon,\,\varepsilon S_{0},\frac{\varepsilon S_{0}^{2}}{2}\big\}$.
Suppose $J\geq3$. Assume density $f$ satisfies Assumption \ref{A1}.
For some even $n\in\mathbb{N}$ define 
\begin{equation}
L=M=\sqrt[n]{\frac{2K\mu_{n}}{\gamma}},\label{eq:M-1}
\end{equation}
where $\mu_{n}$ is the $n^{th}-$moment of $f$, i.e., $\mu_{n}=\frac{1}{i^{n}}\left.\frac{\partial^{n}}{\partial u^{n}}\varphi(u)\right|_{u=0}$.
Let $\xi=\sqrt{2M}K$, $k\in\{1,...,J-2\}$ and 
\begin{align}
N & \geq\left(\frac{2^{k+2}\left(\max_{\ell\in\{0,1,2\}}\left\Vert f^{(k+1+\ell)}\right\Vert _{\infty}\right)L^{k+\frac{3}{2}}}{k\pi^{k+1}}\frac{12\xi}{\gamma}\right)^{\frac{1}{k}}.\label{eq:N_J>0-1}
\end{align}
It follows that the time-0 price and the Greeks Delta and Gamma can
be approximated by the COS method, i.e., 
\begin{align*}
\bigg|\int_{\mathbb{R}}v(x,S_{0})f(x)dx-\sum_{k=0}^{N}{}^{\prime}c_{k}v_{k}\bigg| & \leq\varepsilon,\\
\bigg|\frac{\partial}{\partial S_{0}}\int_{\mathbb{R}}v(x,S_{0})f(x)dx-\bigg(-\frac{1}{S_{0}}\sum_{k=0}^{N}{}^{\prime}c_{k}^{1}v_{k}\bigg)\bigg| & \leq\varepsilon,\\
\bigg|\frac{\partial^{2}}{\partial S_{0}^{2}}\int_{\mathbb{R}}v(x,S_{0})f(x)dx-\frac{1}{S_{0}^{2}}\sum_{k=0}^{N}{}^{\prime}(c_{k}^{1}+c_{k}^{2})v_{k}\bigg| & \leq\varepsilon.
\end{align*}
\end{thm}

\begin{proof}
Theorem \ref{thm:find_N_A_B} ensures that the time-0 price can be
approximated by the COS method. By Lemma \ref{lem:diff_lemma} and
Theorem \ref{thm:find_N_A_B}, it holds that 
\begin{align*}
\bigg|\frac{\partial}{\partial S_{0}}\int_{\mathbb{R}}v(x,S_{0})f(x)dx & -\bigg(-\frac{1}{S_{0}}\sum_{k=0}^{N}{}^{\prime}c_{k}^{1}v_{k}\bigg)\bigg|\\
 & =\frac{1}{S_{0}}\left|\int_{\mathbb{R}}v(x,S_{0})f^{(1)}(x)dx-\sum_{k=0}^{N}{}^{\prime}c_{k}^{1}v_{k}\right|\\
 & \leq\frac{\gamma}{S_{0}}\leq\varepsilon.
\end{align*}
Using the triangle inequality, we can see that 
\begin{align*}
\bigg|\frac{\partial^{2}}{\partial S_{0}^{2}}\int_{\mathbb{R}}v(x,S_{0})f(x)dx & -\frac{1}{S_{0}^{2}}\sum_{k=0}^{N}{}^{\prime}(c_{k}^{1}+c_{k}^{2})v_{k}\bigg|\\
= & \frac{1}{S_{0}^{2}}\bigg|\int_{\mathbb{R}}v(x,S_{0})\big(f^{(1)}(x)+f^{(2)}(x)\big)dx-\sum_{k=0}^{N}{}^{\prime}(c_{k}^{1}+c_{k}^{2})v_{k}\bigg|\\
\leq & \frac{1}{S_{0}^{2}}\bigg(\bigg|\int_{\mathbb{R}}v(x,S_{0})f^{(1)}(x)dx-\sum_{k=0}^{N}{}^{\prime}c_{k}^{1}v_{k}\bigg|\\
 & +\bigg|\int_{\mathbb{R}}v(x,S_{0})f^{(2)}(x)dx-\sum_{k=0}^{N}{}^{\prime}c_{k}^{2}v_{k}\bigg|\bigg)\\
\leq & \frac{2\gamma}{S_{0}^{2}}\leq\varepsilon.
\end{align*}
\end{proof}

\section{\label{sec:Numerical-experiments}Numerical experiments}

Some numerical experiments are compared with the Carr-Madan formula,
see \citet{carr1999option}, which is applicable when the characteristic
function of the log-returns is given in closed form and when $E[S_{T}^{1+\gamma}]$
is finite for some $\gamma>0$, which is the \emph{damping factor}.
Unless otherwise stated, we use the Carr-Madan formula with $N=2^{17}$
terms, we set the damping factor equal to $\gamma=0.1$ and we use
a Fourier truncation range of $1200$ to compute the reference prices.
We implemented the Carr-Madan formula using Simpson's rule without
applying the fast Fourier transform.

All numerical experiments were performed on a modern laptop (Intel
i7-10750H) using the software R and vectorized code without parallelization.

\subsection{\label{subsec:BS-model}BS model}

We consider the BS model with volatility $\sigma>0$ and maturity
$T>0$. The density $f_{\text{BS}}$ of the log-returns is normally
distributed and belongs to the family of stable laws. An upper bound
for $\|f_{\text{BS}}^{(k+1)}\|_{\infty}$ can be obtained directly
from Inequality (\ref{eq:stable_bound}) setting $\alpha=2$ and $c=\frac{\sigma\sqrt{T}}{\sqrt{2}}$.
Let $n\in\mathbb{N}$ be even. In the BS model, the $n^{th}$ moment
of the centralized log-returns is given by
\[
\mu_{n}=(\sigma\sqrt{T})^{n}\left((n-1)(n-3)(n-5)\cdot\cdot\cdot3\cdot1\right).
\]
The formula for $N$ in Inequality (\ref{eq:N_J>0}) does not depend
on the volatility $\sigma\sqrt{T}$ if we set $\ell=0$ and if we
bound $\|f_{\text{BS}}^{(k+1)}\|_{\infty}$ using Inequality (\ref{eq:stable_bound}).
In Figure \ref{fig:N_k}, we show the dependency of $N$ on $k$.
At the beginning, the number of terms $N$ decreases sharply as $k$
increases and stabilizes approximately for $k\geq40$. We also see
that $N$ increases as the error tolerance $\varepsilon$ decreases
or the bound $K$ increases. 

How sharp is the bound for $N$ in Theorem \ref{thm:find_N_A_B} and
Theorem \ref{thm:Greeks}? We make the following experiment. Consider
a put option with parameters
\begin{equation}
\varepsilon=10^{-8},\quad\sigma=0.2,\quad T=1,\quad S_{0}=K=100,\quad r=0.\label{eq:BS}
\end{equation}
We set $n=8$ and $k=40$ to obtain $M=L=6.94$ and $N=218$ by Theorem
\ref{thm:Greeks}. Other values for $k$ and $N$ are reported in
Table \ref{tab:BS}. Theorem \ref{thm:Greeks} indicates that $M,$
$L$ and $N$ serves to approximate both the time-0 price and the
Greeks Delta and Gamma using the COS method. This can be confirmed
by an experiment: $N_{\min}=120$ is the minimal number of terms such
that the absolute differences of the approximation by the COS method
and the closed form solution by the Black-Scholes formula for time-0
price, Delta and Gamma, are less than the error tolerance. $N$ is
about twice as larger as $N_{\min}$.

How can the number of terms $N$ be estimated if there are no closed
form solutions available for the bounds of the derivatives of the
density of the log-returns? We suggest solving the right-hand side
of Inequality (\ref{eq:H_j}) numerically to find a bound for $\|f^{(k+1)}\|_{\infty}$.
Here, we use R's \emph{integrate} function with default values. The
CPU time of the COS method and the numerical integration routine to
obtain a bound for $\|f^{(k+1)}\|_{\infty}$ are of similar magnitude;
see Table \ref{tab:BS}.

The value of $N$ does not change when using numerical integration
to obtain a bound for $\|f^{(k+1)}\|_{\infty}$ compared to the closed
form solution for the bound of $\|f^{(k+1)}\|_{\infty}$.

\begin{table}[H]
\begin{centering}
\begin{tabular}{|>{\centering}p{4cm}|c|c|c|c|c|c|c|}
\hline 
$k$  & 10  & 20  & 30  & 40  & 50  & 60  & 70\tabularnewline
\hline 
$N$  & 988  & 285  & 206  & 183  & 175  & 173  & 174\tabularnewline
\hline 
CPU time: COS method  & 0.19  & 0.09  & 0.09  & 0.07  & 0.08  & 0.07  & 0.08\tabularnewline
\hline 
CPU time: numeric integration for $\|f^{(k+1)}\|_{\infty}$  & 0.15  & 0.09  & 0.09  & 0.10  & 0.09  & 0.10  & 0.09\tabularnewline
\hline 
\end{tabular}
\par\end{centering}
\caption{\label{tab:BS}Approximation of the time-0 price of a put option by
the COS method: $N$ depending on $k$ by Inequality (\ref{eq:N_J>0}).
CPU time is measured in milliseconds. The COS method with $N_{\min}=120$
takes about 0.068 milliseconds. }
\end{table}

\begin{figure}[H]
\begin{centering}
\includegraphics[scale=0.35]{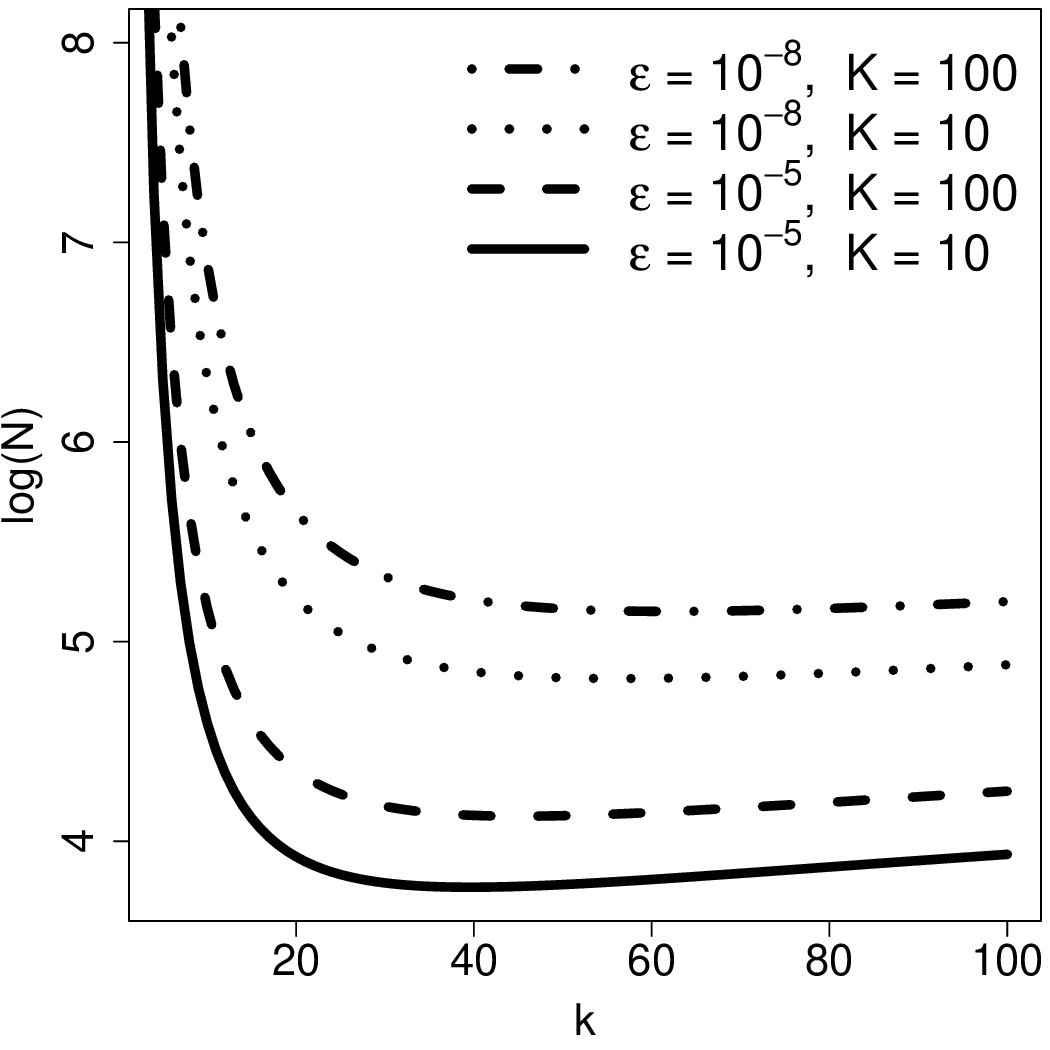}
\par\end{centering}
\caption{\label{fig:N_k}$N$ depending on $k$ by Inequality (\ref{eq:N_J>0})
for various $\varepsilon$ and $K$. We set $n=8$ and $\ell=0$.
In the BS model, $N$ does not depend on the standard deviation of
the log-returns, i.e., on the value of $\sigma\sqrt{T}$.}

\end{figure}

\subsubsection{Last Coefficient Rule-of-Thumb}

By \citet[Sec. 2.12]{boyd2001chebyshev}, the following rule is called
the \emph{Last Coefficient Rule-of-Thumb}: The series truncation error
approximating $f_{L}$ by a finite cosine series is bounded by $\sum_{k>N}|a_{k}|$.
If $N$ is large enough, the order of magnitude of the series truncation
error is expected to scale like $|a_{N}|$. \citet[Sec. 3.2]{aimi2023fast}
propose to select the number of terms for the BS model with volatility
$\sigma>0$ and maturity $T$ by the smallest natural number $N_{F}$
satisfying the following inequality
\begin{equation}
|a_{N_{F}}|\approx|c_{N_{F}}|=\frac{2}{b-a}e^{-\frac{1}{2}(\frac{\sigma\pi}{b-a})^{2}TN_{F}^{2}}\leq\varepsilon_{N_{F}},\label{eq:NF}
\end{equation}
where $[a,b]=[-L,L]$ is the truncation range and $\varepsilon_{N_{F}}$
is some error tolerance. The solution to find the number of terms
by (\ref{eq:NF}) works even better than Inequality (\ref{eq:N_J>0})
if $\varepsilon_{N_{F}}$ is small enough, see Table \ref{tab:Last-Coefficient-Rule-of-Thumb}.
However, the rule by Inequality (\ref{eq:NF}) does not work if $2^{-1}(b-a)\varepsilon_{N_{F}}\geq1$
because we would then choose $N_{F}=0$.

\begin{table}[H]
\begin{centering}
\begin{tabular}{|c|c|c|c|c|c|}
\hline 
$\varepsilon$, $\varepsilon_{N_{F}}$ & $\sigma$ & $L,M$ & $N_{F}$ by (\ref{eq:NF}) & $N$ by (\ref{eq:N_J>0}) & $N_{\min}$\tabularnewline
\hline 
\hline 
$10^{-8}$ & $0.2$ & $6.94$ & $127$ & $183$ & $120$\tabularnewline
\hline 
$0.1$ & $1$ & $10$ & $0$ & $34$ & $10$\tabularnewline
\hline 
\end{tabular}
\par\end{centering}
\caption{\label{tab:Last-Coefficient-Rule-of-Thumb}Last Coefficient Rule-of-Thumb
applied to a put option with $K=S_{0}=100$, $T=1$, $r=0$ and different
volatilities. We choose $k=40$ to obtain $N$ by Inequality (\ref{eq:N_J>0}).}

\end{table}

\subsection{Heston model}

There is an integral expression for the density of the log-returns
in the Heston model and the density is smooth and has semi-heavy tails,
see \citet{dragulescu2002probability}. In this section, we provide
numerical experiments under two different parameter sets, M1 and M2.
The parameters of models M1 and M2 are taken from \citet{fang2009novel}
and \citet{schoutens2003perfect}, respectively. We compute the time-0
prices of put options for different strikes $K$ and different maturities
$T$. Reference prices are obtained by the Carr-Madan formula. 

Table \ref{tab:Heston} compares $N$ obtained by Inequality (\ref{eq:N_J>0})
and the minimal $N_{\min}$ such that the absolute difference of the
approximation by the COS method and the reference price is less than
the error tolerance $\varepsilon$. On average, $N$ is about six
times larger than $N_{\min}$. The CPU time using $N$ instead of
$N_{\min}$ increases roughly by the factor three. To obtain $N$,
we estimate $\|f^{(k+1)}\|_{\infty}$ by solving the right-hand side
of Inequality (\ref{eq:H_j}) by numeric integration using R's function
\emph{integrate }with default values\emph{. }The CPU times of the
COS method and the numerical integration routine are of similar magnitudes.

\begin{table}[H]
\begin{tabular}{|l|c|c|c|c|c|c|c|c|c|c|c|c|}
\hline 
Model & \multicolumn{1}{c}{} & \multicolumn{1}{c}{} & \multicolumn{1}{c}{} & \multicolumn{1}{c}{M1} & \multicolumn{1}{c}{} &  & \multicolumn{1}{c}{} & \multicolumn{1}{c}{} & \multicolumn{1}{c}{} & \multicolumn{1}{c}{M2} & \multicolumn{1}{c}{} & \tabularnewline
\hline 
\hline 
$K$ & 75 & 75 & 100 & 100 & 125 & 125 & 75 & 75 & 100 & 100 & 125 & 125\tabularnewline
\hline 
$T$ & 1 & 2 & 1 & 2 & 1 & 2 & 1 & 2 & 1 & 2 & 1 & 2\tabularnewline
\hline 
$M$, $L$ & 6.3 & 9.7 & 6.8 & 10.4 & 7.2 & 11 & 7.9 & 12.1 & 8.5 & 13.1 & 9 & 13.8\tabularnewline
\hline 
$N$ & 864 & 746 & 948 & 819 & 1019 & 880 & 500 & 560 & 549 & 615 & 590 & 661\tabularnewline
\hline 
$N_{\text{\ensuremath{\min}}}$ & 80 & 110 & 190 & 185 & 110 & 205 & 100 & 100 & 110 & 125 & 145 & 155\tabularnewline
\hline 
CPU COS $N$ & 0.6 & 0.5 & 0.5 & 0.6 & 0.6 & 0.5 & 0.5 & 0.5 & 0.4 & 0.5 & 0.4 & 0.5\tabularnewline
\hline 
CPU COS $N_{\text{\ensuremath{\min}}}$ & 0.2 & 0.2 & 0.2 & 0.3 & 0.2 & 0.2 & 0.2 & 0.2 & 0.2 & 0.2 & 0.2 & 0.2\tabularnewline
\hline 
CPU num. int. & 1.5 & 1.3 & 1.2 & 1.1 & 1.3 & 1.2 & 1.2 & 1.1 & 1.2 & 1.0 & 1.2 & 1.0\tabularnewline
\hline 
\end{tabular}

\caption{\label{tab:Heston}Number of terms $N$ and CPU time to approximate
the time-0 price of a put option by the COS method depending on the
strike $K$, the maturity $T$ and the parameters of the model. We
set $\varepsilon=10^{-3}$, $n=4$, $k=20$, $S_{0}=100$ and $r=0$
to obtain $N$ by Inequality (\ref{eq:N_J>0}). CPU time is measured
in milliseconds. The CPU time to estimate $\|f^{(k+1)}\|_{\infty}$
by numeric integration is also provided in the last row.}

\end{table}

\subsection{\label{subsec:VG-model}VG model}

Using the VG model as an example, this section shows that Theorem
\ref{thm:find_N_A_B} does not help in finding the number of terms
$N$ for the COS method if the density of the log-returns is not sufficiently
smooth. The density of the VG model at time $T>0$ has semi-heavy
tails, see \citet[Example 7.5]{albin2009asymptotic}. It can be expressed
by means of the Whittaker function and is $(J+1)$-times continuously
differentiable if $J+2<\frac{2T}{\nu}$, see \citet{kuchler2008shapes}. 

Let $f_{\text{VG}}$ denote the density of the log-returns in the
VG model at maturity $T>0$ with parameters $\sigma>0$, $\theta\in\mathbb{R}$
and $\nu>0$. If $T<\frac{\nu}{2}$, the density $f_{\text{VG}}$
is unbounded and Theorem \ref{thm:find_N_A_B} cannot be used to find
$N$. If $T\in\left(\nu,\frac{3\nu}{2}\right)$, the derivative of
$f_{\text{VG}}$ is continuous, but the second derivative of $f_{\text{VG}}$
is unbounded, see \citet[Thm. 4.1 and Thm. 6.1]{kuchler2008shapes}.

We can apply Theorem \ref{thm:find_N_A_B} with $J=0$ if $T\in\left(\nu,\frac{3\nu}{2}\right)$,
but the value for $N$ is somewhat useless from a practical point
of view, as the following experiment shows: Consider a European call
option with the following parameters: 
\[
\varepsilon=0.01,\quad\sigma=0.1,\quad\nu=0.2,\quad\theta=0,\quad T=0.25,\quad S_{0}=K=100,\quad r=0.
\]
By Equation (\ref{eq:M}), we set $M=L=0.91$ using $n=4$ moments.
By numerically optimizing the derivative of the density $f_{\text{VG}}$,
we obtain $\|f^{(1)}\|_{\infty}=218$, and Theorem \ref{thm:find_N_A_B}
suggests $N\approx4\cdot10^{14}$.

We calculated a reference price $\pi_{CM}=1.809833$ using the Carr-Madan
formula. Using the COS method, $N=50$ is already sufficient to approximate
the reference price within the error tolerance. 

\subsection{\label{subsec:FMLS-model}FMLS model}

The stable law has been used to model stock prices since \citet{mandelbrot1997variation}
and \citet{fama1965behavior}. A representation of stable densities
by special functions can be found in \citet{zolotarev1995representation}.
The density $f_{\text{Stable}}\in C_{b}^{\infty}(\mathbb{R})$ of
the family of stable distributions with stability parameter $\alpha\in(0,2]$,
skewness $\beta\in[-1,1]$, scale $c>0$ and location $\theta\in\mathbb{R}$
has the characteristic function 
\begin{align*}
u\mapsto & \exp\left(iu\theta-|uc|{}^{\alpha}(1-i\beta\text{sgn}(u)\Phi_{\alpha}(u)\right),\\
 & \Phi_{\alpha}(u)=\begin{cases}
\tan\frac{\pi\alpha}{2} & ,\alpha\neq1\\
-\frac{2}{\pi}\log(|u|) & ,\alpha=1,
\end{cases}
\end{align*}
see \citet{zolotarev1986one}. It follows by Inequality (\ref{eq:H_j})
that 
\begin{align}
\|f_{\text{Stable}}^{(j)}\|_{\infty} & \leq\frac{1}{\pi\alpha c^{j+1}}\int_{0}^{\infty}t^{\frac{j+1}{\alpha}-1}e^{-t}dt\nonumber \\
 & \leq\frac{\Gamma\left(\frac{j+1}{\alpha}\right)}{\pi\alpha c^{j+1}},\quad j=0,1,2,...,\label{eq:stable_bound}
\end{align}
where $\Gamma$ denotes the gamma function. The density of the stable
law is unimodal, see \citet{yamazato1978unimodality}. Let $F_{\text{Stable}}$
be the cumulative distribution function for the stable density $f_{\text{Stable}}$.
By \citet[Property 1.2.15]{samorodnitsky2017stable} it holds that
\begin{align*}
\lim_{x\to\infty}x^{\alpha}(1-F_{\text{Stable}}(x)) & =C_{\alpha}\frac{1+\beta}{2}c^{\alpha},\\
\lim_{x\to\infty}x^{\alpha}(F_{\text{Stable}}(-x)) & =C_{\alpha}\frac{1-\beta}{2}c^{\alpha},
\end{align*}
where 
\[
C_{\alpha}=\begin{cases}
\frac{1-\alpha}{\Gamma(2-\alpha)\cos\big(\frac{\pi\alpha}{2}\big)} & ,\alpha\neq1\\
\frac{2}{\pi} & ,\alpha=1.
\end{cases}
\]
For stable densities we therefore suggest to set $C_{3}$ in Theorem
\ref{thm:find_N_A_B} at least as large as $\alpha C_{\alpha}\frac{1+|\beta|}{2}c^{\alpha}$
to obtain $M$ and $L$. 

The FMLS model with parameters $\sigma>0$ and $\alpha\in(1,2)$ describes
the log-returns by a stable process with maximum negative skewness.
The centralized log-returns in the FMLS model at time $T>0$ are stably
distributed with stability parameter $\alpha$, skewness $\beta=-1$,
scale $c=\sigma T^{\frac{1}{\alpha}}$ and location $\theta=0$.

The density of the log-returns in the FMLS model has a heavy left
tail with Pareto index $\alpha$, i.e., the left tail decays like
$|x|^{-1-\alpha}$, $x\to-\infty$, but the right tail decays exponentially,
see \citet{carr2003finite}. This makes the FMLS very attractive from
a theoretical point of view: put and call option prices and all moments
of the underlying stock $S_{T}$ exist. The expectation of the log-returns
also exists, but the log-returns do not have finite variance. Fitting
the FMLS model to real market data shows very satisfactory results.
\citet{carr2003finite} calibrated the FMLS model to real market data
and estimated $\alpha=1.5597$ and $\sigma=0.1486$. We test the formulas
for $M$, $L$ and $N$ with these values for $\alpha$ and $\sigma$
for a European call option with the following parameters: 
\[
\varepsilon=10^{-2},\quad T=1,\quad S_{0}=K=100,\quad r=0.
\]
The reference price is 9.7433708 by the Carr-Madan formula, and we
confirm the reference price by the COS method with $M=L=10^{5}$ and
$N=10^{7}$.

Figure \ref{fig:Density-of-FMLS} shows the density of the log-returns
in the FMLS model and asymptotic tail behavior, i.e., the function
$x\mapsto C_{3}|x|^{-1-\alpha}$. The left tail does indeed decay
like Pareto tails and the asymptotic tail behavior is very close to
the density. The right tail decays faster; in fact it decays exponentially,
see \citet{carr2003finite}.

\begin{figure}[H]
\begin{centering}
\includegraphics[scale=0.38]{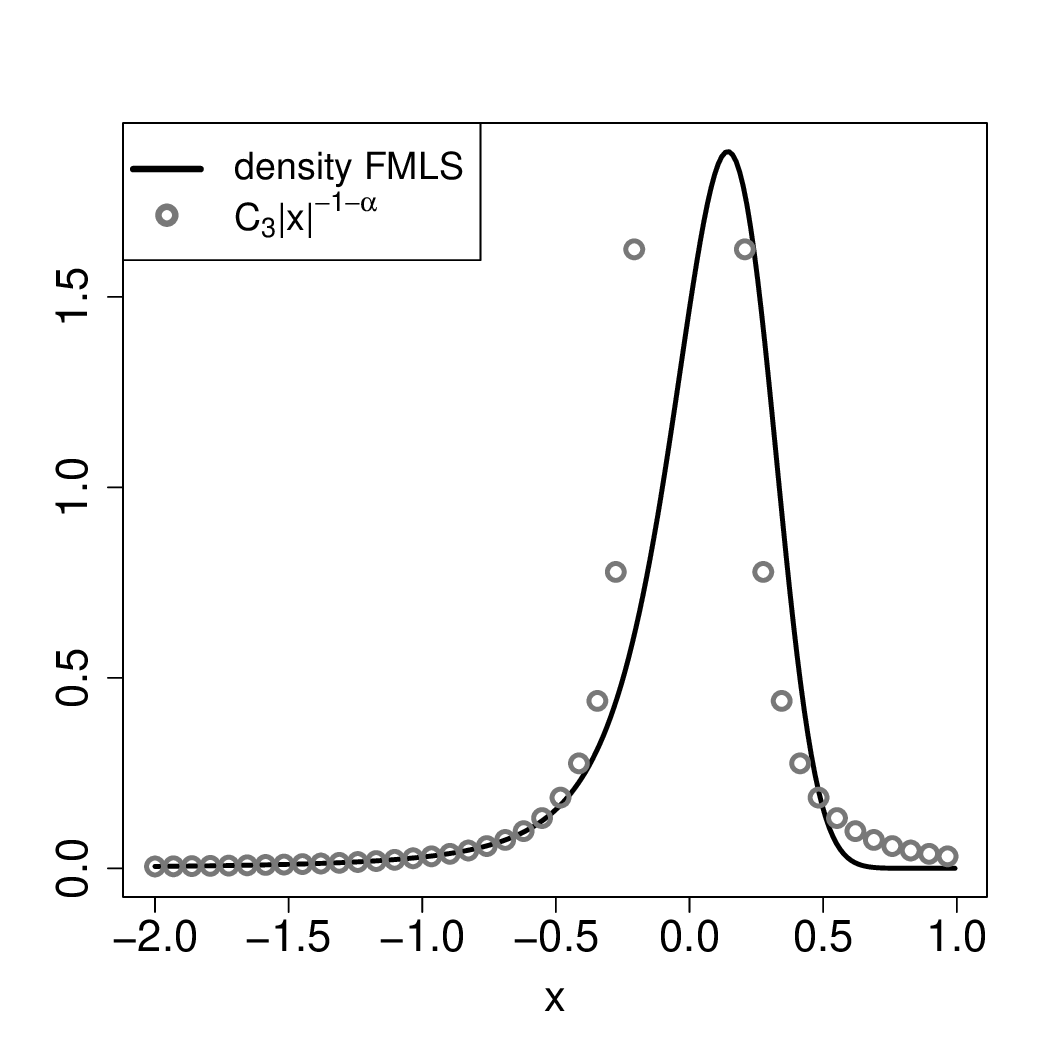} 
\par\end{centering}
\caption{\label{fig:Density-of-FMLS}Density of the log-returns in the FMLS
model and asymptotic tail behavior.}
\end{figure}

To apply the COS method we define by Theorem \ref{thm:find_N_A_B},
$M=69$, $L=176$ and $N=5815$ setting $k=40$. According to Theorem
\ref{thm:find_N_A_B}, any other choice for $k$ is possible, but
another value for $k$ does not significantly improve $N$.

With these parameters, the COS method prices the call option within
the error tolerance in about 1.5 milliseconds. The minimal $N$ to
stay below the error tolerance is $N_{\min}=1200$ and the CPU time
using $N_{\min}$ is about 0.4 milliseconds.

We also apply the Carr-Madan formula with the ``default parameters'',
see \citet{carr1999option}, which are also recommended by \citet[Sec. 3.1]{madan2016applied},
i.e., $4096$ terms, a damping factor equal to $1.5$ and a Fourier
truncation range of $1024$. With these parameters, the Carr-Madan
formula prices the option within the error tolerance in about 1.1
milliseconds. We also implemented the Carr-Madan formula using the
fast Fourier transform but we found no speed advantage, compare also
with \citet{crisostomo2018speed}.

The CPU time is about four times higher when using Theorem \ref{thm:find_N_A_B}
to get $N$ compared to the optimized value for $N$, which is acceptable
from our point of view. The computational time of the COS method using
Theorem \ref{thm:find_N_A_B} and that of the Carr-Madan formula with
standard parameters are of similar magnitude.

\subsection{\label{subsec:Convergence-Rate}Order of Convergence}

In this section, we empirically analyze the order of convergence of
the COS method for the BS model and the FMLS model and compare the
results with Theorem \ref{thm:asy}.

The order of convergence in (\ref{eq:order of con}) can be estimated
straightforwardly by a simulation, see \citet{leisen1996binomial}.
As 
\[
\log\left(\frac{\kappa}{N^{\rho}}\right)=\log(\kappa)-\rho\log(N),
\]
the negative slope of a straight line obtained from a log-log plot
of the error $\text{err}(L(N),N)$ against $N$ can be used as an
indicator for $\rho$. As in Section \ref{sec:Convergence-rate-of},
we assume $M=L$ in this section.

\subsubsection{BS Model}

In the BS model with parameter $\sigma>0$, the density of the log-returns
is arbitrarily smooth, and the tails decay even faster than exponentially.
In Figure \ref{fig:Convergence-rate-BS} we analyze how the error
of the COS method behaves in the BS model for large $N$ and for different
choices of $L$.

We consider a call option with parameters $\sigma=0.2,$ $S_{0}=K=100$,
$r=0$ and $T=1$. We see in Figure \ref{fig:Convergence-rate-BS}
that the COS method seems to converge exponentially at the beginning
for moderate $N$ if we choose $L$ constant, i.e., independent of
$N$. But for constant $L$, e.g., $L=4\sigma$ or $L=6\sigma$, the
COS method does not converge for $N\to\infty$ but the error remains
constant for a certain level of $N$. This can be explained by Inequality
(\ref{eq:error_bound}): While the second term on the right-hand side
of Inequality (\ref{eq:error_bound}) converges to zero for $N\to\infty$
and $L$ fixed, the first and third terms on the right-hand side of
Inequality (\ref{eq:error_bound}) do not improve as $N$ is increased
but $L$ is kept constant.

\begin{figure}[H]
\begin{centering}
\includegraphics[scale=0.38]{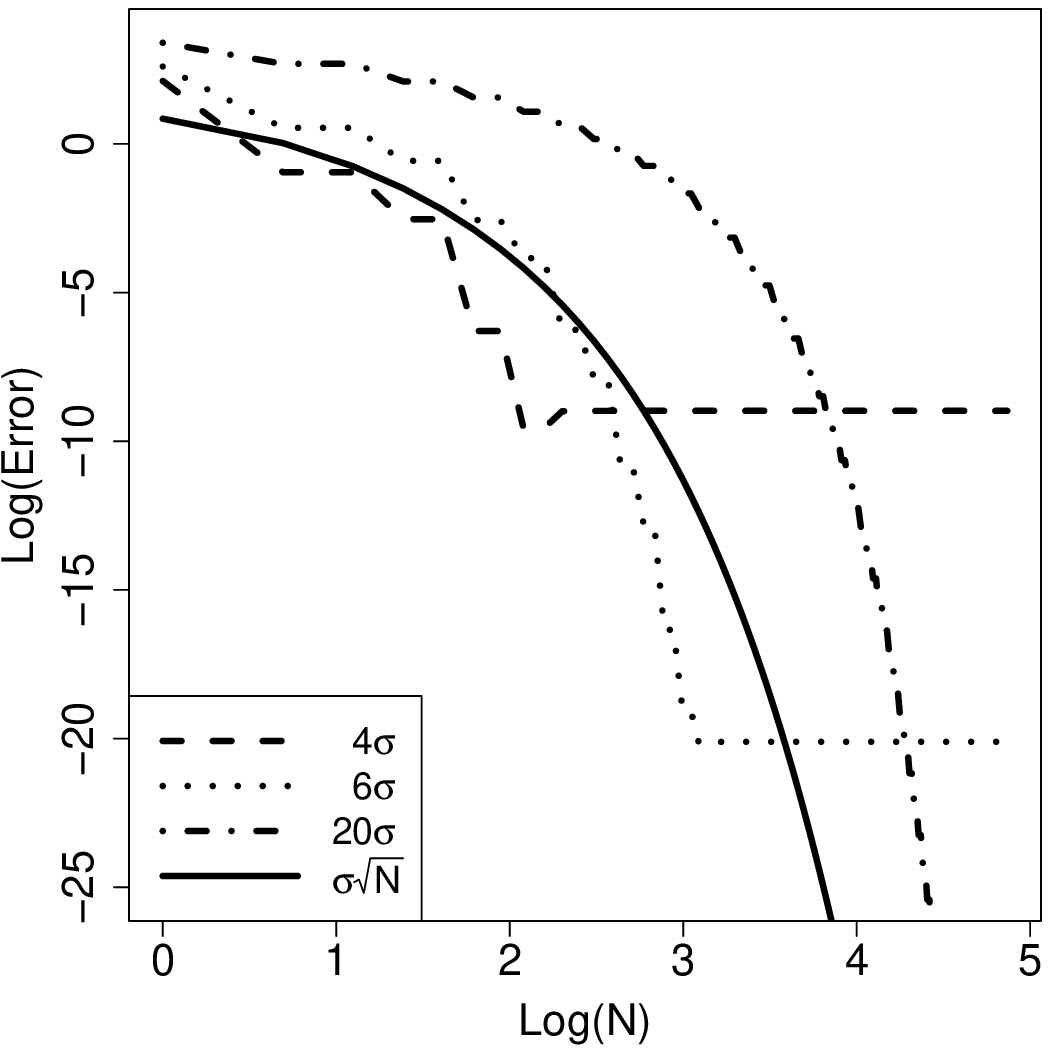} 
\par\end{centering}
\caption{\label{fig:Convergence-rate-BS}Order of convergence of the COS method
for a call option in the BS model with different choices for the truncation
range $[-L,L]$.}
\end{figure}
For large $L$, such as $L=20\sigma$, this effect disappears somewhat
because the tails of the Gaussian distribution decay so rapidly that,
up to fixed-precision arithmetic\footnote{The software package R operates with a precision of 53 bits conforming
to the IEC 60559 standard, see R Core Team (2021). R: A language and
environment for statistical computing. R Foundation for Statistical
Computing, Vienna, Austria. URL https://www.R-project.org/.}, the Gaussian density can be thought of as a density with finite
support. Using arbitrary-precision arithmetic instead should show
that even for $L=20\sigma$, the error of the COS method does not
converge to zero for $N\to\infty$.

If we choose $L=L(N)=\sigma\sqrt{N}$, Theorem \ref{thm:asy}(i) indicates
that the error of the COS method converges exponentially to zero.
This is confirmed empirically in Figure \ref{fig:Convergence-rate-BS}.
Other choices for $L$ also work well, e.g., $L=\frac{\sigma}{5}N$.

\subsubsection{\label{subsec:FMLS-model-1}FMLS model}

We test Theorem \ref{thm:asy}(ii) for the density of the log-returns
in the FMLS model, which belongs to the family of stable densities
and has heavy tails. For the FMLS model we use the parameters $\alpha=1.5597$
and $\sigma=0.1486$: these values are taken from \citet{carr2003finite},
who calibrated the FMLS model to real market data. We use the same
reference price for the time-0 price of a European call option with
maturity $T=1$, $S_{0}=K=100$ and $r=0$ as in Section \ref{subsec:FMLS-model}.

\begin{figure}[H]
\begin{centering}
\includegraphics[scale=0.38]{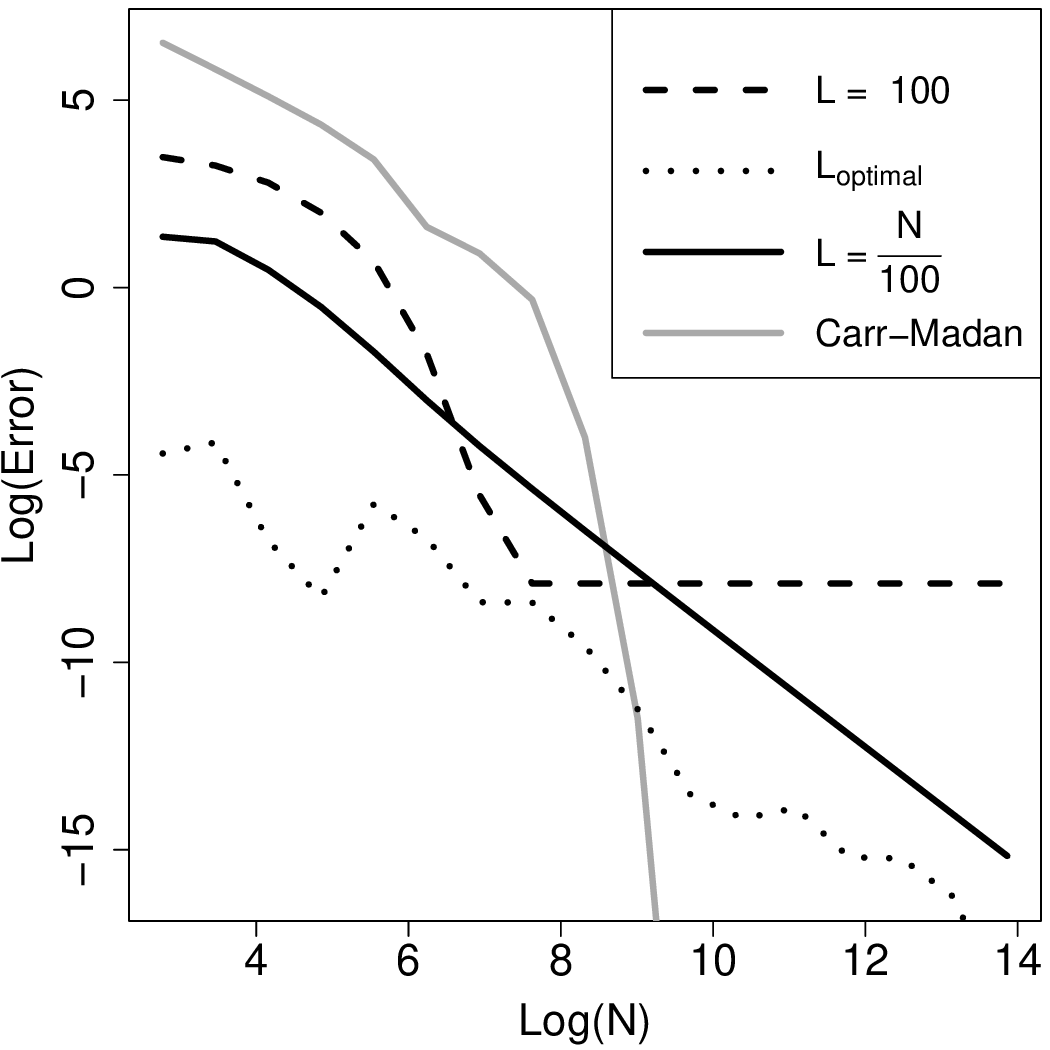} 
\par\end{centering}
\caption{\label{fig:Order-of-convergence}Order of convergence of the COS method
with different choices for the truncation range $[-L,L]$\foreignlanguage{american}{
for the }FMLS model.}
\end{figure}

We compute $L_{\text{optimal}}$ for a fixed $N\in\mathbb{N}$ minimizing
$\text{err}(L,N)$, i.e., for each $N$ we choose the truncation range
such that the error of the COS method is minimal. In particular, for
each $N\in\{2^{m},\,m=4,5,...,20\}$ we define 
\[
L_{\text{optimal}}(N):=\text{argmin}_{L\in\mathbb{L}}\text{err}(L,N),
\]
where 
\[
\mathbb{L}=\{\exp(0.07m),\quad m=0,1,...,200\},
\]
is a sufficiently fine grid of the interval $[1,10^{6}]$.

In Figure \ref{fig:Order-of-convergence}, we compute the order of
convergence of the COS method for different strategies to define $L$.
If $L$ is constant, the COS method does not converge at all.

If we choose $L=\frac{1}{100}N$, the empirical order of convergence
is $1.57$, i.e., the error behaves like $O(N^{-1.57})$ for large
$N$, which is very close to its theoretical bound of $O(N^{-1.5597})$.
The empirical order of convergence does not differ much if we choose
$L_{\text{optimal}}$ instead of $L=\frac{1}{100}N$.

In particular, the numerical experiments indicate that the order of
convergence is not exponential for heavy tail models even though the
corresponding densities are arbitrarily smooth.

In Figure \ref{fig:Order-of-convergence} we also plot the empirical
order of convergence of the Carr-Madan formula for the FMLS model
with damping factor of $0.7$ and a Fourier truncation range of $1200$.
The Figure indicates an exponential order of convergence for the Carr-Madan
formula.

\section{\label{sec:Conclusions}Conclusions}

In this research we analyzed the COS method, which is used for efficient
option pricing. The sensitivities, i.e., the Greeks, can also be efficiently
approximated. The COS method requires two parameters: the truncation
range $[-L,L]$ to truncate the density of the log-returns and the
number of terms $N$ to approximate the truncated density by cosine
functions. We considered stock price models where the density of the
log-returns is smooth and has either semi-heavy tails, i.e., the tails
decay exponentially or faster, or heavy tails, i.e., Pareto tails.

In both cases, we found explicit and useful bounds for $L$ and $N$
and showed in numerical experiments the usefulness of these formulas
in applications to obtain the time-0 price of an option and the Greeks
Delta and Gamma. The densities of the log-returns are smooth for many
models in finance, such as the BS, NIG, Heston and FMLS models.

If the density is not differentiable, Theorem \ref{thm:find_N_A_B}
cannot be used to find a bound for $N$. If the density is only differentiable
a few times, which is the case for the VG model for some parameters
and short maturities, our bound for $N$ is too large to be useful
in most practical applications. 

We further analyzed the order of convergence of the COS method and
observed both theoretically and empirically that the models enjoy
exponential convergence when the densities of the log-returns are
smooth and have semi-heavy tails. However, when the density of the
log-returns is smooth and has \emph{heavy} \emph{tails}, the error
of the COS method can be bounded by $O(N^{-\alpha})$, where $\alpha>0$
is the Pareto tail index. This is the case, for example, for the FMLS
model where $\alpha\in(1,2)$. Numerical experiments indicate that
the bound $O(N^{-\alpha})$ is sharp.

\bibliographystyle{plainnat}
\bibliography{biblio}

\end{document}